\documentclass[11pt]{article}
\usepackage{fullpage}
\usepackage{authblk}
\usepackage{amsthm}
\usepackage{amsfonts}
\usepackage{amsmath}

\newtheorem{theorem}{Theorem}[section]
\newtheorem{corollary}[theorem]{Corollary}

\newtheorem{definition}[theorem]{Definition}
\newtheorem{lemma}[theorem]{Lemma}

\begin{document}
\title{Joint Cache Partition and Job Assignment on Multi-Core Processors\footnote{This work was partially supported by the
Israel Science Foundation grant no. 822-10.
Israeli Centers of Research Excellence (I-CORE) program, (Center  No. 4/11).
Google Inter-university center for Electronic Markets and Auctions.}}

\author[1]{Avinatan Hassidim}
\author[2]{Haim Kaplan}
\author[2]{Omry Tuval}
\affil[1]{Dept. of Computer Science
Bar Ilan University
and Google Israel\\
  \texttt{avinatan@google.com}}
\affil[2]{Dept. of Computer Science, Tel Aviv University\\
  \texttt{\{haimk, omrytuva\}@post.tau.ac.il}}
\renewcommand\Authands{, }
\date{}

\maketitle
\begin{abstract}
Multicore shared cache processors  pose a challenge for designers of
embedded systems  who try to achieve minimal and predictable
execution time of workloads consisting of several jobs.
To address this challenge the cache is statically partitioned among
the cores and the jobs are assigned to the cores so as to minimize the
makespan.
Several heuristic algorithms have been proposed that jointly decide
how to partition the cache among the cores and assign the jobs. We
initiate a theoretical study of this problem which we call the joint
cache partition and job assignment problem.

By a careful analysis of the possible cache partitions we obtain a
constant approximation algorithm for this problem. For some
practical special cases we obtain a $2$-approximation algorithm, and
show how to improve the approximation factor even further by
allowing the algorithm to use additional cache. We also study
possible improvements that can be obtained by allowing dynamic cache
partitions and dynamic job assignments.

We define a natural special case of the well known scheduling
problem on unrelated machines in which machines are ordered by
``strength''. Our joint cache partition and job assignment problem
generalizes this scheduling problem which we think is of independent
interest. We give a polynomial time algorithm for this scheduling
problem for instances obtained  by fixing the cache partition in a
practical case of the joint cache partition and job assignment problem where job loads are step functions.
\end{abstract}

\section{Introduction}

We study the problem of  assigning $n$ jobs to $c$ cores on a multi-core processor, and simultanously partitioning a shared cache of size $K$ among the cores.
Each job $j$ is given by a non-increasing function $T_j(x)$ indicating the running time of job $j$ on a core with cache of size $x$.
A solution is a cache partition $p$, assigning $p(i)$ cache to each core $i$,  and a job assignment $S$ assigning each job $j$ to core $S(j)$.
The total cache allocated to the cores in the solution is $K$, that is  $\sum\limits_{i=1}^{c} p(i) = K$.
The {\em makespan} of a cache partition $p$ and a job assignment $S$ is $\max_i{ \sum_{j|S(j)=i}{T_j(p(i))}}$.
Our goal is to find a cache partition and a job assignment that minimize the makespan.

Multi-core processors are the prevalent  computational architecture
used today in PC's, mobile devices and high performance computing.
Having multiple cores running jobs concurrently, while sharing the
same level 2 and/or level 3 cache, results in complex interactions
between  the jobs, thereby posing a significant challenge in
determining the makespan of a set of jobs.
Cache partitioning has emerged as a technique  to increase run time
predictability and increase performance on multi-core processors
\cite{part_survey, part_throughput}. Theoretic research on online
multi-core caching  shows that  the cache partition (which may be
dynamic)  has more influence on the performance than the eviction
policy \cite{Hassidim,LO12}. To obtain effective cache partitions,
methods have been developed to estimate the running time of jobs as
a function of allocated cache, that is the functions $T_j(x)$ (see
for example the cache locking technique of \cite{liu_lock1}).

Recent empirical research \cite{LZLX,liu_lock2} suggests that
jointly solving for the cache partition among the cores and for the
job assignment to cores leads to significant improvements over combining separate algorithms for the two problems.
The papers \cite{LZLX,liu_lock2} suggest and test heuristic algorithms for the joint cache partition and job assignment problem.
Our work initiates the theoretic study of this problem.

We study this problem in the context of multi-core caching, but our
formulation and results  are applicable in a more general setting,
where  the running time of a job  depends on the availability of
some shared resource (cache, CPU, RAM, budget, etc.) that is allocated to the
machines. This setting is  applicable, for example, for users of a
public cloud infrastructure like Amazon's Elastic Cloud. When a user
decides on her public cloud setup, there is usually a limited
resource (e.g.\ budget), that can be spent on different machines in
the cloud. The more budget is spent on a  machine, it runs jobs
faster and the user is interested in minimizing the makespan of its
set of jobs, while staying within the given budget.

\smallskip

\noindent {\bf Related Work:} Theoretic study of multi-core caching
have shown that traditional online paging algorithms are not
competitive in the multi-core scenario \cite{Hassidim, LO12}. Both
papers \cite{Hassidim, LO12} show that the offline decision version of the caching problem
is NP-complete, in slightly different models.
Much of the difficulty in designing competitive online algorithms for multi-core caching stems from the fact that the way in which the request sequences of the different cores interleave is dependent on the algorithm. An algorithm with good competitive ratio is obtained in \cite{Barve}, when the interleaving of the request sequences is fixed.\\
For related work on scheduling see Section \ref{sec:scheduling}.

\smallskip

\noindent {\bf Our results:} We present a $36$-approximation
algorithm for the joint cache partition and job assignment problem in Section \ref{sec_36}.
We obtain this algorithm by  showing that it suffices to consider a
subset of polynomial size of the cache partitions.

We obtain better approximation guarantees for special cases of the
joint cache partition and job assignment problem.

 When each job has a fixed running time and a minimal cache demand, we present
, in Section \ref{sec_slmc}, a $2$-approximation algorithm, a $\frac{3}{2}$-approximation
algorithm that uses $2K$ cache and a $\frac{4}{3}$-approximation
algorithm that uses $3K$ cache. We call this problem the {\em single load
minimal cache demand} problem.  Our $\frac{4}{3}$-approximation algorithm is based on an algorithm presented in Section \ref{threshold_dominant} that
finds  a dominant perfect matching in a threshold graph that has a perfect matching.
This algorithm and the existence of such a matching are of independent interest.

We present in Section \ref{correl_ptas} a polynomial time approximation scheme  for a special case of the single load minimal cache demand problem,
in which there is a correlation between the jobs' loads and cache demands. Such a model is inspired by practical cases where there is an underlying notion of a job's ``hardness'' that affects both its load and its cache demand.

We study, in Section \ref{sec_const},  the case where the
load functions of the jobs, $T_j(x)$, are step functions. That is, job
$j$ takes $l_j$ time to run if given at least  $x_j$ cache, and
otherwise it takes $h_j > l_j$ time. For the case where there are a constant number of
different $l_j$'s and $h_j$'s we reduce the problem to the single
load minimal cache demand problem and thereby obtain
the same approximation results as for that problem (Section \ref{sec_const}).

We define the problem of scheduling on {\em ordered unrelated
machines}, a natural special case of the classical job scheduling
problem on unrelated machines.
In this problem there is a total order on the machines which captures their relative strength.
Each job has a different running time on each machine
and these running times are non-increasing with the strength of the
machine. We show a reduction from this problem to the joint
cache partition and job assignment problem. We also give a
polynomial time dynamic programming algorithm for a special case of
this problem that arises when we fix the cache partition in the
special case where the number of $l_j$s and $h_j$s is constant (Section \ref{sec_const}).

In section \ref{sec_variants} we generalize the joint cache
partition and job assignment problem and consider dynamic cache
partitions and dynamic job schedules. We show upper and lower bounds
on the makespan improvement that can be gained by using dynamic partitions and
dynamic assignments.

\section{The ordered unrelated machines problem}\label{sec:scheduling}
The ordered unrelated machines scheduling problem is defined as follows.
There are $c$ machines and a set $J$ of jobs.
The input is a matrix $T(i,j)$ giving the running time of job $j$ on machine $i$, such that for each two machines $i_1<i_2$ and any job $j$, $T(i_1,j) \ge  T(i_2,j)$. The goal is to assign the jobs to the machines such that  the makespan is minimized.

The ordered unrelated machines scheduling problem is a  special case of scheduling on unrelated machines in which
there is a total order on the machines that captures their relative strengths.
This special case is natural since
in many practical scenarios the machines have some underlying notion of strength and jobs run faster on a stronger machine.
For example a newer computer typically dominates an older one in all parameters, or a more experienced employee does any job faster than
a new recruit.

Lenstra et al \cite{LST90}
gave a $2$ approximation algorithm for scheduling on unrelated machines based on
 rounding an optimal fractional
solution to a linear program, and proved that it is NP-hard to approximate the problem to within a factor better than
$\frac{3}{2}$. Shchepin and  Vakhania \cite{SV05} improved Lenstra's rounding technique  and obtained a $2-\frac{1}{c}$ approximation algorithm. It is currently an open question if there are better approximation algorithms for ordered unrelated machines than the more general algorithms that approximate unrelated machines.

Another well-studied scheduling problem is scheduling on uniformly related machines. In this problem, the time it takes for machine $i$ to run job $j$ is $\frac{l_j}{s_i}$ where $l_j$ is the  load of job $j$ and $s_i$ is the  speed of machine $i$. A polynomial time approximation scheme for related machines is described in \cite{HS88}. It is easy to see that the problem of scheduling on related machines  is a special case of the problem of scheduling on ordered unrelated machines, and therefore the ordered unrelated machines problem is NP-hard.

The ordered unrelated machines problem is closely related to the joint cache partition and job assignment problem.
Consider an instance of the joint cache partition and job assignment problem with $c$ cores, $K$ cache and a set of jobs $J$ such that $T_j(x)$ is the load function of job $j$.
If we fix the cache partition to be some arbitrary partition $p$, and we index the cores in non-decreasing order of their cache allocation,
then we get an instance of the ordered unrelated machines problem, where $T(i,j)=T_j(p(i))$.
Our constant approximation algorithm for the joint cache partition and job assignment problem, described in Section \ref{sec_36}, uses this observation as well as Lenstra's $2$-approximation for unrelated  machines.
In the rest of this section we prove that the joint cache partition and job assignment problem is at least as hard as the
ordered unrelated machines scheduling problem.

We reduce  the ordered unrelated machine problem to the joint cache partition and job assignment problem.
Consider the decision version of the ordered unrelated scheduling problem, with $c$ machines and $n=|J|$ jobs, where job $j$ takes time $T(i,j)$ to run on machine $i$.
We want to decide if it is possible to schedule the jobs on the machines with makespan at most $M$.

Define the following instance of the joint cache partition and job assignment problem.
This instance has  $c$ cores, a total cache $K = c(c+1)/2$ and $n' = n + c$ jobs.
The first $n$ jobs ($1\le j \le n$) correspond to the jobs in the original ordered unrelated machines problem, and $c$ jobs are new jobs ($n+1\le j\le n+c$).
The load function $T_j(x)$ of job $j$, where $1 \le j\le n$,  equals $T(x,j)$ if $x\leq c$ and equals $T(c,j)$ if $x>c$.
The load function $T_j(x)$ of job $j$, where $n+1 \le j\le n+c$,  equals $M+\delta$ if $x\geq j-n$ for some $\delta>0$ and equals $\infty$ if $x<j-n$.
Our load functions $T_j(x)$ are non-increasing because the original $T(i,j)$'s are non-increasing in the machine index $i$.

\begin{lemma} \label{reduce}
The makespan of the joint cache partition and job assignment instance defined above is at most $2M+\delta$ if and only if the makespan of the original unrelated scheduling problem is at most $M$.
\end{lemma}
\begin{proof}
Assume there is an assignment $S'$ of the jobs in the original ordered  unrelated machines instance of makespan at most $M$.
We show a cache partition $p$ and job assignment $S$ for the joint cache partition and job assignment instance with makespan at most $2M+\delta$.

The cache partition $p$ is defined such that $p(i)=i$ for each core $i$. The partition $p$ uses exactly $K=c(c+1)/2$ cache.
The job  assignment $S$ is defined such that for
 a job $j>n$, $S(j)=j-n$ and for a job $j\leq n$, $S(j)=S'(j)$. The partition $p$ assigns $i$ cache to
core $i$, which is exactly enough for job $n+i$, which is assigned to core $i$ by $S$, to run in time $M+\delta$.
It is easy to verify that $p,S$ is a solution to the joint cache partition and job assignment instance with makespan at most $2M+\delta$.

Assume there is a solution  $p,S$ for the joint cache partition and job assignment instance, with makespan at most $2M+\delta$.
Job $j$, such that $n < j \le n+c$, must run on a core with cache at least $j - n$, or else the makespan would be infinite.
Moreover, no two jobs $j_1>n$ and $j_2>n$ are assigned by $S$ the same core, as this would give a makespan of at least $2M + 2 \delta$.
Combining these observations with the fact that the total available cache is $K=c(c+1)/2$, we get that the cache partition must be $p(i)=i$ for each core $i$.
Furthermore, each job $j>n$ is assigned by $S$ to core $j-n$ and all the other jobs assigned by $S$ to core $j-n$ are jobs
corresponding to original jobs in the ordered unrelated machines instance. Therefore, the total load of original jobs assigned by $S$ to core $i$ is at most $M$.

We define $S'$, a job assignment for the original ordered unrelated machines instance, by setting $S'(j)=S(j)$ for each $j\leq n$.
Since $S$ assigns original jobs of total load at most $M$ on each core, it follows that the makespaen of $S'$ is at most $M$.
\end{proof}

The following theorem follows immediately from Lemma \ref{reduce}
\begin{theorem}
There is a polynomial-time reduction from the ordered unrelated machines scheduling problem to the joint cache partition and job assignment problem.
\end{theorem}

The reduction in the proof of Lemma \ref{reduce}  does not preserve approximation guarantees.
However by choosing $\delta$ carefully we can get the following result.
\begin{theorem}\label{reduce_approx}
  Given an algorithm $A$ for the joint cache partition and job assignment problem that approximates the optimal makespan up to a factor of $1 + \epsilon$, for $0<\epsilon<1$, we can construct an algorithm for
  the ordered unrelated machines scheduling problem that approximates the optimal makespan up to a factor of $1 + 2\epsilon + \frac{2\epsilon^2}{1 - \epsilon - \chi}$ for any $\chi > 0$.
\end{theorem}
\begin{proof}
We first obtain a $\left( 1+ 2\epsilon + \frac{2\epsilon^2}{1-\epsilon-\chi}\right)$-approximation algorithm for the decision version of the ordered unrelated machines scheduling problem.
That is, an algorithm that given a value $M$, either decides that there is no assignment of makespan $M$ or finds an assignment with makespan $\left( 1+ 2\epsilon + \frac{2\epsilon^2}{1-\epsilon-\chi}\right) M $.

Given an instance of the ordered unrelated machines scheduling problem, we construct an instance of the joint cache partition and job assignment as described before lemma \ref{reduce},
and set  $\delta = \frac{2 \epsilon M }{1 - \epsilon - \chi}$, for an arbitrarily small $\chi>0$.
We use algorithm $A$ to solve the resulting instance of the joint cache partition and job assignment problem. Let $p,S$ be the solution returned by $A$.
We define $S'(j) = S(j)$ for each $1\le j \le n$.
If the makespan of $S'$ is at most $\left(1 + 2\epsilon + \frac{2\epsilon^2}{1 - \epsilon - \chi}\right)M$ we return $S'$ as the solution and otherwise decide that there is no solution with makespan at most $M$.

If the makespan of the original instance is at most $M$, then  by lemma \ref{reduce} there is a solution to the joint cache partition and job assignment instance resulting from the reduction, with makespan at most $2M+\delta$.
Therefore  $p,S$, the solution returned by algorithm $A$, is of makespan at most $(1+\epsilon)(2M+\delta)$.

By our choice of $\delta$ we have that $(1 + \epsilon) (2M + \delta) < 2M + 2 \delta$ and therefore each core is assigned by $S$ at most one job $j$, such that $j>n$.
In addition, any job $j$ such that $n < j \le n+c$, must run on a core with cache at least $j - n$, or else the makespan would be infinite.
Combining these observations with the fact that the total available cache is $K=c(c+1)/2$, we get that the cache partition must be $p(i)=i$ for each core $i$.
Furthermore, each job $j>n$ is assigned by $S$ to core $j-n$ and all the other jobs assigned by $S$ to core $j-n$ are jobs
corresponding to original jobs in the ordered unrelated machines instance. Therefore, the total load of original jobs assigned by $S$ to core $i$ is at most $(1 + \epsilon) (2M + \delta) - M - \delta$.
It follows that the makespan of $S'$ is
 at most $(1 + \epsilon) (2M + \delta) - M - \delta = M \left( 1+ 2\epsilon + \frac{2\epsilon^2}{1-\epsilon-\chi}\right)$.

We obtained a $\left( 1+ 2\epsilon + \frac{2\epsilon^2}{1-\epsilon-\chi}\right)$-approximation algorithm for the decision version of the ordered unrelated machines scheduling problem.
In order to approximately solve the optimization problem, we can perform a binary search for the optimal makespan using the approximation algorithm for the decision version of the problem and get a
$\left( 1+ 2\epsilon + \frac{2\epsilon^2}{1-\epsilon-\chi}\right)$-approximation algorithm for the optimization problem. We obtain an initial search range for the binary search by using $\sum\limits_{j=1}^n T(c,j) $ as an upper bound on the makespan of the optimal schedule and $\frac{1}{c} \sum\limits_{j=1}^n T(c,j)$ as a lower bound.
(See section \ref{apn_approx_des_to_opt} for a detailed discussion of a similar case of using an approximate decision algorithm in a binary search framework to obtain an approximate optimization algorithm.)
\end{proof}

\section{A constant approximation algorithm}\label{sec_36}

We first obtain an 18-approximation algorithm that for the joint cache partition and job assignment problem that uses
$(1+\frac{5}{2}\epsilon)K$ cache for some constant  $0 < \epsilon <\frac{1}{2}$.
We then show another algorithm that uses $K$ cache and approximates the makespan up to a factor of 36.

Our first algorithm, denoted by $A$, enumerates over a  subset of
cache partitions, denoted by $P(K,c,\epsilon)$. For each partition
in this set $A$ approximates the makespan of  the corresponding
scheduling problem, using Lenstra's  algorithm, and returns the
partition and associated job assignment with the smallest makespan.

Let $K'=(1+\epsilon)^{\lceil \log_{1+\epsilon}(K)\rceil}$, the smallest integral power of $(1+\epsilon)$ which is at least $K$.
The set $P(K,c,\epsilon)$ contains cache partitions in which the cache allocated to each core is
an integral power of $(1+\epsilon)$ and the number of different integral powers used by the partition is at most $\log_2(c)$.
We denote by $b$ the number of different cache sizes in a partition.
Each core is allocated $\frac{K'}{(1+\epsilon)^{l_j }}$ cache, where $l_j\in \mathbb{N}$ and $1\le j \le b$.
The smallest possible cache allocated to any core is the smallest integral power of $(1+\epsilon)$ which is at least $\frac{K\epsilon}{c}$ and the largest possible cache allocated to a core is $K'$.
We denote by $\hat{\sigma}_j$ the number of cores with cache at least $\frac{K'}{(1+\epsilon)^{l_j }}$.
It follows that there are $(\hat{\sigma}_j - \hat{\sigma}_{j-1})$ cores with  $\frac{K'}{(1+\epsilon)^{l_j}}$ cache.
We require that $\hat{\sigma}_j$ is an integral power of $2$ and that the total cache used is at most $\left(1+\frac{5}{2}\epsilon\right)K$.
Formally,
\begin{alignat}{2}
P(K,c,\epsilon)=\{ (l=&<l_1,\ldots,l_b>,\hat{\sigma}=<\hat{\sigma}_0,\hat{\sigma}_1,\ldots,\hat{\sigma}_b>) \mid \; b\in\mathbb{N}, \quad 1\leq b\leq\log_2{c}  \\
& \forall j,   l_j\in\mathbb{N}, 0\leq l_j \leq \log_{1+\epsilon}\left(\frac{c}{\epsilon}\right)+1, \quad \forall j, \;  l_{j+1}>l_j \\
&\forall j \; \exists u_j\in\mathbb{N} \quad s.t. \quad \hat{\sigma}_{j}=2^{u_j},  \hat{\sigma}_0=0, \hat{\sigma}_b \leq c, \quad \forall j \; \hat{\sigma}_{j+1} > \hat{\sigma_j}  \\
 &\sum\limits_{j=1}^{b} (\hat{\sigma}_{j}-\hat{\sigma}_{j-1})
\frac{K'}{(1+\epsilon)^{l_j }} \leq
\left(1+\frac{5}{2}\epsilon\right)K\}
\end{alignat}
When the parameters are clear from the context, we use $P$ to denote
$P(K,c,\epsilon)$. Let $M(p,S)$ denote the makespan of cache
partition $p$ and job assignment $S$. The following theorem
specifies the main property of $P$, and is proven in the remainder of this section.
\begin{theorem}\label{P9thm}
Let $p,S$ be any cache partition and job assignment. There are
a cache partition and a job assignment $\hat{p},\hat{S}$ such that
$\hat{p}\in P$ and $M(\hat{p},\hat{S})\leq9M(p,S)$.
\end{theorem}

An immediate corollary of Theorem  \ref{P9thm} is that algorithm $A$
described above finds a cache partition and job assignment with
makespan  at most 18 times the optimal makespan.

Lemma \ref{lem_poly} shows that $A$ is a polynomial time algorithm.

\begin{lemma}\label{lem_poly}
The size of $P$ is polynomial in $c$.
\end{lemma}
\begin{proof}
Let  $(l,\hat{\sigma})\in P$.
The vector $\hat{\sigma}$ is a strictly increasing vector of integral powers of $2$, where each power is at most $c$. Therefore the number of possible vectors for $\hat{\sigma}$ is bounded by the number of subsets of $\{2^0,\ldots,2^{\log_2(c)}\}$ which is $O(2^{\log_2{c}})=O(c)$.
The vector $l$ is a strictly increasing vector of integers, each integer is at most $\log_{1+\epsilon}({\frac{c}{\epsilon}}) + 1$.
Therefore the number of vectors $l$ is bounded by the number of subsets of integers that are at most $\log_{1+\epsilon}(\frac{c}{\epsilon}) + 1$ which is $O(2^{\log_{1+\epsilon}(\frac{c}{\epsilon})})=O(2^{\frac{\log_2(\frac{c}{\epsilon})}{\log_2(1+\epsilon)}})=Poly(c)$ since $\epsilon$ is a constant.
Therefore $|P|= O(c \, 2^{\frac{\log_2(\frac{c}{\epsilon})}{\log_2(1+\epsilon)}})$.
\end{proof}

Let $(p,S)$ be a cache partition and a job assignment that use $c$ cores, $K$ cache and have a makespan $M(p,S)$.
Define a cache partition $p_1$ such that for each core $i$, if $p(i)< \frac{K\epsilon}{c}$ then $p_1(i)=\frac{K\epsilon}{c}$ and if $p(i)\geq\frac{K\epsilon}{c}$ then $p_1(i)=p(i)$. For each core $i$, $p_1(i)\leq p(i) + \frac{k\epsilon}{c}$ and hence the total amount of cache allocated by $p_1$ is bounded by $(1+\epsilon)K$. For each core $i$, $p_1(i)\geq p(i)$ and therefore $M(p_1,S)\leq M(p,S)$.

Let $p_2$ be a cache partition such that for each core $i$, $p_2(i)=(1+\epsilon)^{\lceil \log_{1+\epsilon}(p_1(i))\rceil}$, the smallest integral power of $(1+\epsilon)$ that is at least $p_1(i)$.
For each $i$, $p_2(i)\geq p_1(i)$ and thus $M(p_2,S)\leq M(p_1,S) \leq M(p,S)$.  We increased the total cache allocated by at most a multiplicative factor of $(1+\epsilon)$ and therefore the total cache used by $p_2$ is at most $(1+\epsilon)^2 K \leq (1+\frac{5}{2}\epsilon)K$ since $\epsilon < \frac{1}{2}$.

Let $\varphi$ be any  cache partition that allocates  to each core
an integral power of $(1+\epsilon)$ cache. We define the notion of
$\textit{cache levels}$. We say that core $i$ is \textit{of cache
level $l$} in $\varphi$ if $\varphi(i)=
\frac{K'}{(1+\epsilon)^{l}}$. Let $c_l(\varphi)$ denote the number
of cores in cache level $l$ in $\varphi$.
The vector of $c_l$'s, which we call the \textit{cache levels
vector} of $\varphi$, defines the partition $\varphi$ completely
since any two partitions that have the same cache level vector are
identical up to a renaming of the cores.

Let $\sigma(\varphi)$ be the vector of prefix sums of the cache
levels vector of $\varphi$. Formally, $\sigma_l (\varphi) =
\sum\limits_{i=0}^{l} c_i(\varphi)$. Note that $\sigma_l(\varphi)$
is the number of cores in cache partition $\varphi$ with at least
$\frac{K'}{(1+\epsilon)^{l}}$ cache and that for each $l$,
$\sigma_l(\varphi) \leq c$.

For each such cache partition  $\varphi$, we define
the \textit{significant cache levels} $l_i(\varphi)$ recursively as
follows. The first significant cache level $l_1(\varphi)$ is the
first cache level $l$ such that $c_l(\varphi)>0$. Assume we already
defined the $i-1$ first significant cache levels and let
$l'=l_{i-1}(\varphi)$ then $l_{i}(\varphi)$ is the smallest cache level
$l > l'$ such that $\sigma_{l}(\varphi) \geq 2\sigma_{l'}(\varphi)$.

\begin{lemma}\label{consec}
Let $l_j$ and $l_{j+1}$ be two consecutive significant cache levels
of $\varphi$, then the total number of cores in cache levels in
between $l_j$ and $l_{j+1}$  is at most $\sigma_{l_j}(\varphi)$. Let
$l_b$ be the last significant cache level of $\varphi$ then the total
number of cores in cache levels larger than $l_b$ is at most
 $\sigma_{l_b}(\varphi)$.
\end{lemma}
\begin{proof}
Assume to the contrary that  $\sum\limits_{f=l_j +1}^{l_{j+1} - 1}
c_f(\varphi) \geq \sigma_{l_j}(\varphi)$. This implies that for $l'=l_{j+1}-1$, $\sigma_{l'}(\varphi) \geq 2\sigma_{l_j}(\varphi)$ which contradicts the assumption that there are no significant cache
levels in between $l_j$ and $l_{j+1}$ in $\varphi$. The proof of the
second part of the lemma is analogous.
\end{proof}

Let $c_l=c_l(p_2)$. For each core $i$,  $ \frac{K\epsilon}{c} \leq
p_2(i)\leq K'$, so we get that if $l$ is a cache level in $p_2$ such
that $c_l\neq 0$ then $0 \leq l \leq
\log_{1+\epsilon}(\frac{c}{\epsilon}) + 1$. Let
$\sigma_l=\sigma_l(p_2)$ and $\sigma = <
\sigma_1,\ldots,\sigma_{b'}>$, where $b' =
\log_{1+\epsilon}(\frac{c}{\epsilon}) + 1$.
 Let $l_i = l_i(p_2)$, for $1\leq i \leq b$, where $b$ is the number of
significant cache levels in $p_2$.

 We adjust $p_2$ and $S$ to create a new cache
partition $p_3$ and a new job assignment $S_3$. Cache partition
$p_3$  has cores only in the significant cache levels $l_1,\ldots,l_b$
of $p_2$. We obtain $p_3$ from $p_2$ as follows. Let $f$ be a
non-significant cache level in $p_2$. If there is a $j$ such that
$l_{j-1} < f < l_j$ then we take the $c_f$ cores  in cache level $f$ in $p_2$
and reduce their cache so they are now in cache level $l_j$ in $p_3$. If $f
> l_b$ then we remove the $c_f$ cores at level $f$ from our solution.
It is easy to check that the significant cache levels of $p_3$ are
the same as of $p_2$, that is  $l_1,\ldots,l_b$. Since we only reduce the cache allocated to some cores, the new
cache partition $p_3$ uses no more cache than $p_2$ which is at most $(1+\frac{5}{2}\epsilon)K$.

We construct $S_3$ by changing the assignment of the jobs assigned
by $S$ to cores in non-significant cache levels in $p_2$. As before, let $f$
be a nonsignificant cache level and let $l_{j-1}$ be the maximal
significant cache level such that $l_{j-1}<f$. For each core $i$ in
cache level $f$ in $p_2$ we move all the jobs assigned by $S$ to
core $i$, to a target core in cache level $l_{j-1}$ in $p_3$. Lemma
\ref{move_jobs_lem} specifies the key property of this
job-reassignment.

\begin{lemma}\label{move_jobs_lem}
We can construct  $S_3$ such that each core  in a significant level
of $p_3$ is the target of the jobs from at most two cores in a
nonsignificant level of $p_2$.
\end{lemma}
\begin{proof}
Let $c^3$ denote the cache levels vector of $p_3$ and let
$\sigma^3$ denote the vector of prefix sums of $c^3$.
From the definition of  $p_3$ follows that for all $j$,
$\sigma^3_{l_j} = \sigma_{l_j}$, and that for $j > 1$,
$c^3_{l_j}= \sigma^3_{l_j}-\sigma^3_{l_{j-1}} =
\sigma_{l_j}-\sigma_{l_{j-1}}$.

By Lemma \ref{consec} the number of cores in nonsignificant levels
in $p_2$ whose jobs are reassigned to one of the $c^3_{l_j}$
cores in level $l_j$ in $p_3$ is at most $\sigma_{l_j}$. So for $j >
1$ the ratio between the number of cores whose jobs are reassigned
to the number of target cores in level $l_j$ in $p_3$ is at most
$\frac{\sigma_{l_j}}{\sigma_{l_j}-\sigma_{l_{j-1}}} = 1 +
\frac{\sigma_{l_{j-1}}} {\sigma_{l_j}-\sigma_{l_{j-1}}} \leq 2$. For
$j=1$ the number of target cores in level $l_1$ of $p_3$ is
$c^3_{l_1}=\sigma_{l_1}$ which is at least as large as the
number of cores at nonsignificant levels between $l_1$ and $l_2$ in
$p_2$ so we can reassign the jobs of a single core of a
nonsignificant level between $l_1$ and $l_2$ in $p_2$ to each target
core.
\end{proof}

\begin{corollary}\label{crl_first_3}
$M(p_3,S_3)\leq 3M(p,S)$
\end{corollary}
\begin{proof}
In the new cache partition $p_3$ and job assignment $S_3$ we have
added to each core at a significant level in $p_3$ the jobs from at
most $2$ other cores at nonsignificant levels in $p_2$. The target
core always has more cache than the original core, thus the added
load from each original core is at most $M(p_2,S)$. It follows that
$M(p_3,S_3) \leq 3M(p_2,S)\leq3M(p,S)$.
\end{proof}

Let $c^3$ denote the cache levels vector of $p_3$ and let
$\sigma^3$ denote the vector of prefix sums of $c^3$. We now define another
cache partition $\hat{p}$ based on $p_3$. Let $u_j = \lfloor
\log_2(\sigma^3_{l_j})\rfloor$. The partition $\hat{p}$ has
$2^{u_1}$ cores in cache level $l_1$, and $2^{u_j}-2^{u_{j-1}}$
cores in cache level $l_j$ for $1<j\le b$. The cache levels
$l_1,\ldots,l_b$ are the significant cache levels of $\hat{p}$ and
$\hat{p}$ has cores only in its significant cache levels. Let
$\hat{c}_{l_j}$ denote the number of cores   in the significant cache
level $l_j$ in  $\hat{p}$.

\begin{lemma} \label{pppp}
 3$\hat{c}_{l_j} \geq c^3_{l_j}$
\end{lemma}
\begin{proof}
By the definition of $u_j$, we have that $2^{u_j} \leq
\sigma^3_{l_j} < 2^{u_j + 1}$. So for $j>1$
\begin{equation}\label{ratio_eq}
\frac{\hat{c}_{l_j}}{c^3_{l_j}} = \frac{2^{u_j} -
2^{u_{j-1}}}{\sigma^3_{l_j} - \sigma^3_{l_{j-1}}} >
\frac{2^{u_j} - 2^{u_{j-1}}}{2^{u_j + 1} - 2^{u_{j-1}}} =
\frac{2^{u_j - u_{j-1}} - 1}{2^{u_j - u_{j-1} + 1} - 1} \
\end{equation}
Since $l_j$ and $l_{j-1}$ are two consecutive significant cache
levels we have that $u_j - u_{j-1} \geq 1$. The ratio in
\ref{ratio_eq} is an increasing function of $u_j - u_{j-1}$ and thus
 minimized by $u_j - u_{j-1} = 1$, yielding a lower bound of
$\frac{1}{3}$. For $j=1$,
$\frac{\hat{c}_{l_1}}{c^3_{l_1}} =
\frac{2^{u_1}}{\sigma^3_{l_1}} > \frac{2^{u_1}}{2^{u_1 + 1}} =
\frac{1}{2}$.
\end{proof}

Lemma \ref{pppp} shows that the cache partition  $\hat{p}$ has in
each cache level $l_j$ at least a third of the cores that $p_3$ has at
level $l_j$. Therefore, there exists a job assignment $\hat{S}$ that
assigns to each core of cache level $l_j$ in $\hat{p}$  the jobs that
$S_3$ assigns to at most $3$ cores in cache level $l_j$ in $p_3$. We
only moved jobs within the same cache level and thus their load
remains the same, and the makespan $M(\hat{p},\hat{s}) \leq
3M(p_3,S_3) \leq 9M(p,s)$.

\begin{lemma}\label{lem_hatinp}
Cache partition $\hat{p}$ is in the  set $P(K,c,\epsilon)$.
\end{lemma}
\begin{proof}
Let $\hat{\sigma}$  be the vector of prefix sums of $\hat{c}$. The
vectors
$<l_1,\ldots,l_b>,<\hat{\sigma}_{l_1},\ldots,\hat{\sigma}_{l_b}>$
clearly satisfy properties 1-3 in the definition of
$P(K,c,\epsilon)$. It remains to show that $\hat{p}$ uses at most
$(1+\frac{5}{2}\epsilon)K$ cache (property 4).

Consider the core  with the $x$th largest cache in $\hat{p}$. Let
$l_j$ be the cache level of this core. Thus $\hat{\sigma}_{l_j} \geq
x$. Since  $\hat{\sigma}_{l_j}$ is the result of rounding down
$\sigma^3_{l_j}$ to the nearest integral power of $2$, we have
that $\hat{\sigma}_{l_j} \leq \sigma^3_{l_j}$. It follows that
$\sigma^3_{l_j} \geq x$ and therefore the core with the $x$th
largest cache in $p_3$ is in cache level $l_j$ or smaller and thus
is it has at least as much cache as the $x$th largest core in
$\hat{p}$. So $\hat{p}$ uses at most the same amount of cache as
$p_3$ which is at most $(1+\frac{5}{2}\epsilon)K$.
\end{proof}

This concludes the proof of Theorem \ref{P9thm}, and establishes
that our algorithm $A$ is an 18-approximation algorithm for the
problem, using $(1+\frac{5}{2}\epsilon)K$ cache.

We provide a variation of algorithm $A$ that uses at most $K$ cache, and finds a $36$-approximation for  the optimal makespan.
Algorithm $B$ enumerates on $r$, $1\leq r \leq K$, the amount of cache allocated to the first core.
It then enumerates over the set of partitions $P=P(\frac{K-r}{2},\lceil \frac{c}{2}\rceil - 1,\frac{2}{5})$.
For each partition in $P$ it adds another core with $r$ cache and applies Lenstra's approximation algorithm on the resulting instance of the unrelated machines scheduling problem, to assign all the jobs in $J$ to the $\lceil \frac{c}{2}\rceil$ cores. Algorithm $B$ returns the partition and assignment with the minimal makespan it encounters.

\begin{theorem}\label{thm_36final}
If there is a solution of makespan $M$ that uses at most $K$ cache and at most $c$ cores then algorithm $B$ returns a solution of makespan $36M$ that uses at most $K$ cache and at most $c$ cores.
\end{theorem}
\begin{proof}
Let $(p,S)$ be the a solution of makespan $M$, $K$ cache and $c$ cores. W.l.o.g. assume that the cores
are indexed according to the non-increasing order of their cache allocation in this solution, that is $p(i+1)>p(i)$.

Let $J'=\{j\in J \mid S(j)\ge 3\}$.
Consider the following job assignment $S'$ of the jobs in $J'$ to the cores of odd indices greater than $1$ in $(p,S)$/
The assignment $S'$ assigns to core $2i-1$, for $i\ge 2$, all the jobs that are assigned by $S$ to cores
$2i-1$ and $2i$.
Note that all the jobs
assigned by $S'$ to some core are assigned by $S$ to a core with at most the same amount of cache and thus
the makespan of $S'$ is at most $2M$.

Assume $r=p(1)$. Then $K = r + \sum\limits_{odd\,i\geq 3} p(i)+p(i-1) \geq r + \sum\limits_{odd\,i\geq 3} 2p(i)$ since $p$ is non-increasing.
Therefore  we get that $\sum\limits_{odd\,i\geq 3} p(i) \leq \frac{K-r}{2}$.
Therefore we can assign the jobs in $J'$ to $\lceil \frac{c}{2}\rceil - 1$ cores with a total cache of $\frac{K-r}{2}$, such that the  makespan is at most $2M$.
By Theorem \ref{P9thm}, there is a partition $\hat{p}' \in P(\frac{K-r}{2},\lceil \frac{c}{2}\rceil - 1,\frac{2}{5})$ that allocates at most $(1+\frac{5}{2}\frac{2}{5})\frac{K-r}{2}=K-r$ cache to $\lceil \frac{c}{2}\rceil - 1$ cores, and a job assignment $\hat{S}'$ of the jobs in $J'$ to these cores such that the makespan of $\hat{p}',\hat{S}'$ is at most $18M$.

Let $\hat{p}$ be a cache partition that adds to $\hat{p}'$  another core (called ``core 1'') with $r$ cache.
The total cache used by $\hat{p}$ is at most $K$.
Let $\hat{S}$ be a job assignment such that $\hat{S}(j) = \hat{S}'(j)$ for $j\in J'$ and
for a job $j \in J \setminus J'$ (a job that was assigned by $S$ either to core $1$ or to core $2$), $\hat{S}(j) = 1$.
Since the makespan of $(p,S)$ is $M$ we know that the load on core $1$ in the solution $\hat{p}, \hat{S}$ is at most $2M$. It follows that the makespan
of $\hat{p}, \hat{S}$ is at most $18M$.

When algorithm $B$ fixes the size of the cache of the first core to be $r=p(1)$, and
considers $\hat{p}' \in P(\frac{K-r}{2},\lceil \frac{c}{2}\rceil - 1,\frac{2}{5})$ then it obtains the cache partition
$\hat{p}$. We  know that $\hat{S}$ is a solution to the corresponding scheduling problem with makespan at most $18M$.
Therefore Lenstra's approximation algorithm finds an assignment with makespan at most $36M$.
\end{proof}

\section{Jobs with a single load and a minimal cache demand}\label{sec_slmc}

We consider a special case of the general joint cache partition and
job assignment  problem where each job has a minimal cache demand
$x_j$ and single load value $a_j$. Job $j$ must run on a core with
at least $x_j$ cache and it contributes a load of $a_j$ to the core.
We want to decide if the jobs can be assigned to  $c$ cores, using $K$ cache, such that the
makespan is at most $m$? W.l.o.g. we assume $m=1$.

In Section \ref{slmc_2} we describe a $2$-\textit{approximate decision} algorithm that if the given instance has a solution of makespan at most $1$, returns a solution with makespan at most $2$ and otherwise may fail. In Sections \ref{slmc_32} and \ref{slmc_43}  we improve the approximation guarantee to $\frac{3}{2}$ and $\frac{4}{3}$ at the expense of using $2K$ and $3K$ cache, respectively.
In Section \ref{apn_approx_des_to_opt} we show how to obtain an approximate optimization algorithm using an approximate decision algorithm and a standard binary search technique.

\subsection{2-approximation}\label{slmc_2}
We present a 2-approximate decision algorithm, denoted by $A_2$.
Algorithm $A_2$ sorts the jobs in a non-increasing order of their cache
demand. It then assigns the jobs to the cores in this order. It
keeps assigning jobs to a core until the load on the core exceeds
$1$. Then, $A_2$ starts assigning jobs to the next core.
Note that among the jobs assigned to a specific core the first one is the
most cache demanding and it determines the cache allocated to this core by $A_2$.
Algorithm $A_2$ fails if the generated solution uses more than $c$ cores or more than $K$ cache.
Otherwise, $A_2$ returns the generated cache partition and job assignment.

\begin{theorem}\label{joint2}
If there is a cache partition and job assignemtn of makespan at most $1$ that use $c$ cores and $K$ cache then algorithm $A_2$ finds a cache partition
and job assignment of makespan at most $2$ that use at most $c$ cores and at most $K$ cache.
\end{theorem}
\begin{proof}
Let $Y=(p,S)$ be  the cache partition and job assignment  with
makespan $1$ whose existence is assumed by the lemma. $Y$ has makespan $1$ so the sum of the loads of all jobs is at most
$c$. Since $A_2$ loads each core, except maybe the last one, with more
than $1$ load it follows that $A_2$ uses at most $c$ cores.

Since $Y$ has makespan $1$ the load of each of the jobs is at most
$1$. Algorithm $A_2$ only exceeds a load of $1$ on a core by the load
of the last job assigned to this core and thus $A_2$ yields a
solution with makespan at most $2$.

Assume w.l.o.g that the cores in $Y$ are
indexed such that for any core $i$, $p(i+1)\leq p(i)$. Assume that
the cores in $A_2$ are indexed in the order in which they were loaded
by $A_2$. By the definition of $A_2$ the cores are also sorted by
non-increasing order of their cache allocation. Denote by $z(i)$ the
amount of cache $A_2$ allocates to core $i$. We show that for
all $i\in\{1,\ldots,c\}$, $z(i) \leq p(i)$. This implies that
algorithm $A_2$ uses at most $K$ cache.

$A_2$ allocates to the first core the cache required by the most
demanding job so $z(1)=\max_j{x_j}$. This job must be assigned in
$Y$ to some core and therefore $z(1) \leq p(1)$. Assume to the
contrary that $z(i)
> p(i)$ for some $i$. Each job $j$ with cache demand $x_j
>p(i)$ must be assigned in $Y$ to one of the first $(i-1)$ cores,
because all the other cores don't have enough cache to run this
job. Since $Y$ has makespan $1$ we know that $\sum\limits_{j \mid
x_j>p(i)}a_j \leq (i-1)$. Consider all the jobs with cache demand at
least $z(i)$. Algorithm $A_2$ failed to assign all these jobs to the first
$(i-1)$ cores, and we know that $A_2$ assigns more than $1$ load to
each core. So $\sum\limits_{j \mid x_j\geq z(i)}a_j > (i-1)$. Since
$z(i)>p(i)$ and there is a job with cache demand $z(i)$, we have $\sum\limits_{j \mid x_j\geq z(i)}a_j <
\sum\limits_{j|x_j>p(i)}a_j$ which leads to a contradiction.
Therefore $z(i) \leq p(i)$ for all $i$ and  algorithm $A_2$ uses at
most $K$ cache.
\end{proof}

\subsection{$\frac{3}{2}$-approximation with $2K$ cache}\label{slmc_32}
We define a job to be {\em large} if $a_j
> \frac{1}{2}$ and {\em small} otherwise.
Our algorithm $A_\frac{3}{2}$ assigns one large job to each core.
 Let
$s_i$ be the load on core $i$ after the large jobs are assigned. Let
$r_i=1-s_i$. We process the small jobs by non-increasing order
of their cache demand $x_j$, and assign them to the cores in
non-increasing order of the cores' $r_i$'s.
We stop assigning  jobs to a core when its load exceeds 1 and start loading the next core.
Algorithm $A_\frac{3}{2}$ allocates to each core the cache demand of its
most demanding job. Algorithm $A_{\frac{3}{2}}$ fails if the resulting solution uses more than $c$ cores or more than $2K$ cache.

\begin{theorem}\label{proof_32}
If there is a cache partition and job assignment of makespan at most $1$ that use $c$ cores and $K$ cache then $A_\frac{3}{2}$ finds a cache partition and job assignment that use at most $2K$ cache, at most $c$ cores and have a makespan of at most $\frac{3}{2}$.
\end{theorem}
\begin{proof}
Let $Y=(p,S)$ be  the cache partition and job assignment  with makespan $1$ whose existence is assumed by the lemma. The existence of $Y$ implies that there are at most $c$ large jobs in our input and that the total volume of all the jobs is at most $c$. Therefore algorithm  $A_\frac{3}{2}$ uses at most $c$ cores to assign the large jobs. Furthermore, when  $A_\frac{3}{2}$ assigns the small jobs it loads each core, except maybe the last one, with a load of at least $1$ and thus uses at most $c$ cores.
Algorithm  $A_\frac{3}{2}$ provides a solution with makespan at most $\frac{3}{2}$ since it can only exceed a load of $1$ on any core by the load of a single small job.

Let $z$ be the cache partition generated by  $A_\frac{3}{2}$. Let $C_l$ be the set of cores whose most cache demanding job is a large job and $C_s$ be the set of cores whose most cache demanding job is a small job.
For core $i\in C_l$, Let $j_i$ be the most cache demanding job assigned to core $i$, so we have $z(i)=x_{j_i}$. The solution $Y=(p,S)$ is a valid solution thus $x_{j_i} \leq p(S(j_i))$ so $z(i) \leq p(S(j_i))$. If $j_1,j_2$ are two large jobs then $S(j_1)\neq S(j_2)$ and we get that $\sum\limits_{i\in C_l} z(i) \leq \sum\limits_{i\in C_l} p(S(j_i)) \leq \sum\limits_{i=1}^{c} p(i) = K$.

In the rest of the proof we index the cores in the solution of  $A_\frac{3}{2}$ such that $r_1 \geq r_2 \ldots \geq r_c$.This is the same order in which   $A_\frac{3}{2}$ assigns small jobs to the cores. In $Y$ we assume that the cores are indexed such that $p(i)\geq p(i+1)$.
We now prove the $z(i)\leq p(i)$ for any core $i\in C_s$. Assume, to the contrary, that for some $i$, $z(i) > p(i)$. Let $\alpha$ be the cache demand of the most cache demanding small job on core $i$ in $Y$.
Let $J_1= \{j \mid a_j\leq\frac{1}{2}, x_j \geq z(i)\}$ and let $J_2= \{j \mid a_j\leq\frac{1}{2}, x_j > \alpha)\}$. Since $\alpha \leq p(i)$ and by our assumption $p(i)<z(i)$ we get that $\alpha <z(i)$ and therefore $J_1 \subseteq J_2$.

$A_\frac{3}{2}$ does not assign all the jobs of $J_1$ to its first $(i-1)$ cores and therefore the total load of the jobs in $J_1$ is greater than  $\sum\limits_{l=1}^{i-1}r_l$.
On the other hand we know that in $Y$, assignment $S$ assigns all the jobs in $J_2$ on its first $i-1$ cores while not exceeding a load of 1. Thus the total load of jobs in $J_2$ is at most the space available for small jobs on the first $(i-1)$ cores in solution $Y$.
Since $r_1 \geq r_2 \ldots \geq r_c$,  and since in any solution each core runs at most one large job, we get that $\sum\limits_{l=1}^{i-1}r_l$ is at least as large as the space available for small jobs in any subset of $(i-1)$ cores in any solution.
It follows that the total load of jobs in $J_2$ is smaller than in $J_1$. This contradicts the fact that $J_1 \subseteq J_2$.

We conclude that for every $i\in C_s$, $z(i)\leq p(i)$. This implies that the total cache allocated to cores in $C_s$ is at most $K$. We previously showed that the total cache allocated to cores in $C_l$ is at most $K$ and thus the total cache used by algorithm $A_\frac{3}{2}$ is at most $2K$.

\end{proof}

\subsection{$\frac{4}{3}$-approximation with $3K$ cache, using dominant matching}\label{slmc_43}
We present a $\frac{4}{3}$ approximate decision algorithm, $A_\frac{4}{3}$, that uses at most $3K$ cache.
The main challenge is assigning the \textit{large jobs}, which here are defined as jobs of load greater than $\frac{1}{3}$.

There are at most $2c$ large jobs in our instance, because we assume there is a solution of makespan at most $1$ that uses $c$ cores.
Algorithm $A_\frac{4}{3}$ matches these large jobs into pairs, and assigns each pair to a different core.
In order to perform the matching, we construct a graph $G$ where each vertex represents a large job $j$ of weight $a_j>\frac{1}{3}$. If needed, we add artificial vertices of weight zero to have a total of exactly $2c$ vertices in the graph.  Each two vertices have an edge between them if the sum of their weights is at most $1$. The weight of an edge is the sum of the weights of its endpoints.

A perfect matching in a graph is a subset of edges such that every vertex in the graph is incident to exactly one edge in the subset. We note that there is a natural bijection between perfect matchings  in the graph $G$ and assignments of makespan at most $1$ of the large jobs to the cores.
The $c$ edges in any perfect matching define the assignment of the large jobs to the $c$ cores as follows: Let $(a,b)$ be an edge in the perfect matching.
If both $a$ and $b$ correspond to large jobs, we assign both these jobs to the same core.
If $a$ corresponds to a large job and $b$ is an artificial vertex, we assign the job corresponding to $a$ to its own core.
If both $a$ and $b$ are artificial vertices, we leave a core without any large jobs assigned to it. Similarly we can injectively map any assignment of the larges  jobs of makespan at most $1$ to a perfect matching in $G$: For each core that has 2 large jobs assigned to it, we select the edge in $G$ corresponding to these jobs, for each core with a single large job assigned to it, we select an edge between the corresponding real vertex and an arbitrary artificial vertex, and for each core with no large jobs assigned to it we select an edge in $G$ between two artificial vertices.

 A \textit{dominant perfect matching} in $G$ is a perfect matching $Q$ such that for every $i$, the $i$ heaviest edges in $Q$ are a maximum weight matching in $G$ of $i$ edges.
The graph $G$ is a threshold graph \cite{MP95}, and in Section \ref{threshold_dominant} we provide a polynomial time algorithm that finds a dominant perfect matching in any threshold graph that has a perfect matching. If there is a solution for the given instance of makespan at most $1$ then the assignment of the large jobs in that solution correspond to a perfect matching in $G$ and thus algorithm $A_\frac{4}{3}$ can apply the algorithm from Section \ref{threshold_dominant} and find a dominant perfect matching, $Q$, in $G$.

Algorithm $A_\frac{4}{3}$ then assigns the small jobs (load $ \leq \frac{1}{3}$) similarly to algorithms $A_2$ and $A_\frac{3}{2}$ described in  Sections \ref{slmc_2} and \ref{slmc_32}, respectively.
It greedily assigns jobs to a core, until the core's load exceeds $1$. Jobs are assigned in a non-increasing order of their cache demand and the algorithm goes through the cores in a non-decreasing order of the sum of loads of the large jobs on each core. Once all the jobs are assigned, the algorithm allocates cache to the cores according to the cache demand of the most demanding job on each core.
Algorithm $A_\frac{4}{3}$ fails if it does not find a dominant perfect matching in $G$ or if the resulting solution uses more than $c$ cores or more than $3K$ cache.

\begin{theorem}\label{proof_43}
If there is a solution that assigns the jobs to $c$ cores with makespan $1$ and uses $K$ cache then algorithm $A_\frac{4}{3}$ assigns the jobs to $c$ cores with makespan at most $\frac{4}{3}$ and uses at most $3K$ cache.
\end{theorem}
\begin{proof}
Let $Y=(p,S)$ be a solution of makespan at most $1$, that uses $c$ cores and $K$ cache.

Algorithm  $A_\frac{4}{3}$ provides a solution with makespan at most $\frac{4}{3}$ since it may only exceed a load of $1$ on any core by the load of a single small job.

Algorithm  $A_\frac{4}{3}$ uses at most $c$ cores to assign the large jobs because the assignment is based on a perfect matching of size $c$ in $G$.
 The existence of $Y$ implies that the total load of all jobs is at most $c$. When $A_\frac{4}{3}$ assigns the small jobs it exceeds a load of $1$ on all cores it processes, except maybe the last one, and therefore we get that $A_\frac{4}{3}$ uses at most $c$ cores.

Let $z$ be the cache partition generated by  $A_\frac{4}{3}$. Let $C_l$ be the set of cores whose most demanding job is a large job and $C_s$ be the set of cores whose most demanding job is a small job.

Consider any core $i\in C_l$. Let $j$ be the most cache demanding large job assigned to core $i$. Job $j$ runs in solution $Y$ on some core $S(j)$. Therefore $z(i)=x_j\leq p(S(j))$. Since each core in $Y$ runs at most two large jobs, we get that the total cache allocated by our algorithm to cores in $C_l$ is at most $2K$.

Consider the large jobs assigned to cores according to the dominant perfect matching $Q$. Denote by $s_i$ the load on core $i$ after the large jobs are assigned (and before the small jobs are assigned) and let $r_i=1-s_i$.
W.l.o.g. we assume the cores in $A_\frac{4}{3}$ are indexed such that $r_1\geq\ldots\geq r_c$.
For every $i$, $\sum\limits_{l=i}^{c} s_l$ is at least as large than this sum in any assignment of the large jobs of makespan at most $1$ because any such assignment defines a perfect matching in graph $G$ and if $\sum\limits_{l=i}^{c} s_l$ is larger in some other assignment then $Q$ is not a dominant perfect matching in $G$.
Since the total volume of all large jobs is fixed, we get that for every core $i$ the amount of free volume on cores $1$ till $i$, $\sum\limits_{l=1}^{i} r_l$, is maximal and can not be exceeded by any other assignment of the large jobs of makespan at most $1$.

W.l.o.g we assume that the cores in solution $Y=(p,S)$ are indexed such that $p(i)\geq p(i+1)$.
Let $i$ be any core in $C_s$. We show that $z(i)\leq p(i)$. Assume, to the contrary, that $z(i) > p(i)$.
Let $\alpha$ be the cache demand of the most demanding small job assigned to core $i$ in solution $Y$.
Let $J_1=\{j \mid a_j\leq\frac{1}{3}, x_j\geq z(i)\}$ and $J_2=\{j\mid a_j\leq\frac{1}{3},x_j > \alpha\}$. Since  $\alpha \leq p(i) < z(i)$, we get that $J_1 \subseteq J_2$.

Solution $Y$ assigns all the jobs in $J_2$ to its first $(i-1)$ cores, without exceeding a makespan of $1$. Therefore the total volume of jobs in $J_2$ is at most the total available space solution $Y$ has on its first $(i-1)$ cores after assigning the large jobs. Since we know that for every $i$, $\sum\limits_{l=1}^{i} r_l$ is maximal and can not be exceeded by any assignment of the large jobs of makespan at most $1$, we get that the total volume of jobs in $J_2$ is at most  $\sum\limits_{l=1}^{i} r_l$. Algorithm $A_\frac{4}{3}$ does not assign all the jobs in $J_1$ to its first $(i-1)$ cores, and since $A_\frac{4}{3}$ loads each of the first $(i-1)$ cores with at least $1$, we get that the total volume of jobs in $J_1$ is greater than $\sum\limits_{l=1}^{i} r_l$. So we get that the total volume of jobs in $J_2$ is less than the total volume of jobs in $J_1$ but that is a contradiction to the fact that $J_1\subseteq J_2$. Therefore we get that $z(i)\leq p(i)$, for every $i\in C_s$.  It follows that the total cache allocated by our algorithm to cores in $C_s$ is at most $K$ and this concludes the proof that our algorithm allocates a total of at most $3K$ cache to all cores.
\end{proof}

\subsection{Approximate optimization algorithms for the single load, minimal cache model }\label{apn_approx_des_to_opt}
We presented approximation algorithms for the decision version of the joint cache partition and job assignment problem in the single load and minimal cache demand model.
If there is a solution with makespan $m$, algorithms $A_2$, $A_\frac{3}{2}$ and $A_\frac{4}{3}$ find a solution of makespan $2m$, $\frac{3m}{2}$ and $\frac{4m}{3}$, that uses $K$, $2K$ and $3K$ cache, respectively.
We now show how to transform these algorithms into approximate optimization algorithms using a standard binary search technique \cite{LST90}.

\begin{lemma}\label{bin_opt}
Given $m$, $K$ and $c$, assume there is a polynomial time approximate decision algorithm that if there is a solution of makespan $m$, $K$ cache and $c$ cores, returns a solution of makespan $\alpha m$, $\beta K$ cache and $c$ cores,  where $\alpha$ and $\beta$ are at least $1$. Then, there is a polynomial time approximation algorithm that finds a solution of makespan $\alpha m_{opt}$, $\beta K$ cache and $c$ cores, where $m_{opt}$ is the makespan of the optimal solution with $K$ cache and $c$ cores.
\end{lemma}

\begin{proof}
Let's temporarily assume that the loads of all jobs are integers.
This implies that for any cache partition and job assignment the makespan is an integer.

Our approximate optimization algorithm performs a binary search for the optimal makespan and maintains a search range $[L,U]$.
Initially, $U= \sum\limits_{j=1}^{n}a_j $ and $L= \left \lceil \frac{1}{c} U \right \rceil$.
Clearly these initial values of $L$ and $U$ are a lower and an upper bound on the optimal makespan, respectively.
Let $A$ be the approximate decision algorithm whose existence is assumed in the lemma's statement.
In each iteration, we run algorithm $A$ with parameters $K$, $c$ and $m=\lfloor \frac{L+U}{2} \rfloor$.
If $A$ succeeds and returns a solution with makespan at most $\alpha m$ we update the upper bound $U:=m$.
If $A$ fails, we know there is no solution of makespan at most $m$, and we update the lower bound $L:=m+1$.
It is easy to see that the binary search maintains the invariant that after any iteration, if the search range is $[L,U]$ then $m_{opt}\in[L,\alpha U]$ and we have a solution of makespan at most $\alpha U$. The binary search stops when $L=U$.

The makespan of the solution when the binary search stops is at most $\alpha U=\alpha L \leq \alpha m_{opt}$. The binary search stops after $O(\log_2(\sum\limits_{j=1}^{n}a_j))$ iterations, and since $A$ runs in polynomial time, we get that our algorithm runs in polynomial time.
This shows that our binary search algorithm is a polynomial time $\alpha$-approximation algorithm.

If the loads in our instance are not integers, let $\frac{1}{2^\phi}$ be the precision in which the loads are given.
By multiplying all loads by $2^\phi$ we get an equivalent instance where all the loads of the jobs are integers.
Note that this only adds $\phi$ iterations to the binary search and our algorithm still runs in polynomial time.
\end{proof}

The following theorem follows immediately from Lemma \ref{bin_opt}.

\begin{theorem}
Using the approximate decision algorithms presented in this section, we obtain polynomial time approximate optimization algorithms for the single load, minimal cache demand problem with approximation factors $2$,\,$\frac{3}{2}$ and $\frac{4}{3}$ that use $K$, $2K$ and $3K$ cache,  respectively.
\end{theorem}

\subsection{Dominant perfect matching in threshold graphs}\label{threshold_dominant}

Let $G=(V,E)$ be an undirected graph with $2c$ vertices where each vertex $x\in V$ has a weight $w(x) \geq 0$. The edges in the graph are defined by a threshold $t>0$ to be $E=\{ (x,y) \mid w(x)+w(y) \leq t, x\neq y\}$. Such a graph $G$ is known as a threshold graph \cite{CH73,MP95}. We say that the \textit{weight} of an edge $(x,y)$ is $w(x,y)=w(x)+w(y)$.

A perfect matching $A$ in $G$ is a subset of the edges such that every vertex in $V$ in incident to exactly one edge in $A$.
Let $A_{i}$ denote the $i$-th heaviest edge in $A$. We assume, w.l.o.g, that there is some arbitrary predefined order of the edges in $E$ that is used, as a secondary sort criteria, to break ties in case several edges have the same weight. In particular, this implies that $A_{i}$ is uniquely defined.

\begin{definition}
A perfect matching $A$ dominates a perfect matching $B$ if for every $x\in\{1,\ldots,c\}$ $\sum\limits_{i=1}^{x} w(A_{i}) \geq \sum\limits_{i=1}^{x} w(B_{i})$
\end{definition}

\begin{definition}
A perfect matching $A$ is a dominant matching if $A$ dominates any other perfect matching $B$.
\end{definition}

Let $A$ and $B$ be two perfect matchings in $G$.
We say that $A$ and $B$ \textit{share a prefix of length $l$} if $A_{i}=B_{i}$ for $i \in \{1,\ldots,l\}$.
The following greedy algorithm finds a dominant perfect matching in a threshold graph $G$ that has a perfect matching.
We start with $G_0=G$.
At step $i$, the algorithm selects the edge $(x,y)$ with maximum weight in the graph $G_i$. If there are several edges of maximum weight, then $(x,y)$ is the first by the predefined order on $E$.
The graph $G_{i+1}$ is obtained from $G_i$ by removing vertices $x$, $y$ and all edges incident to $x$ or $y$.
The algorithm stops when it selected $c$ edges and $G_c$ is empty.

\begin{lemma}\label{not_stuck}
For every $x\in\{0,\ldots,c-1\}$, If graph $G_x$ has a perfect matching, then the graph $G_{x+1}$ has a perfect matching.
\end{lemma}
\begin{proof}
Let $M_x$ denote the perfect matching in graph $G_x$.
Let $(a,b)$ be the edge of maximum weight in $G_x$ that we remove, with its vertices and their incident edges, to obtain $G_{x+1}$.
If $(a,b)\in M_x$ then clearly $M_x\setminus \{ (a,b)\}$ is a perfect matching in $G_{x+1}$.
If $(a,b)\not \in M_x$, and since $M_x$ is a perfect matching of $G_x$, there are two vertices $c$ and $d$ such that $(a,c)$ and $(b,d)$ are in $M_x$.
The edge $(a,b)$ is the maximum weight edge in $G_x$ and thus $w(b)\geq w(c)$ and $w(a) \geq w(d)$. Therefore $(c,d)$ must be an edge in $G_x$ because $w(c)+w(d) \leq w(a)+w(b) \leq t$ the threshold defining the edges in our threshold graph.
Let $M_{x+1} = M_x  \setminus \{(a,c),(b,d)\} \cup \{ (c,d)\}$.
It is easy to see that $M_{x+1}$ is a perfect matching of graph $G_{x+1}$.
\end{proof}

\begin{theorem}\label{dominant_matching}
If $G$ is a threshold graph with $2c$ vertices that has a perfect matching, then the greedy algorithm described above finds a dominant perfect matching.
\end{theorem}
\begin{proof}
Lemma \ref{not_stuck} implies that our greedy algorithm is able to select a set of $c$ edges that is a perfect matching in $G$. Denote this matching by $Q$.

Assume, to the contrary, that $Q$ is not a dominant perfect matching in $G$.
Let $A$ be a perfect matching that is not dominated by $Q$ sharing the longest possible prefix with $Q$. Let $x$ denote the length of the shared prefix of $Q$ and $A$. Let $G_x$ denote the graph obtained from $G$ by removing the $x$ edges that are the heaviest in both $A$ and $Q$, their vertices and all edges incident to these vertices.

Let $(a,b)=Q_{x+1}$. Since $A$ and $Q$ share a maximal prefix of length $x$, $A_{x+1}\neq(a,b)$ .
Since $(a,b)$ is of maximum weight in $G_x$, it follows that $(a,b)\not \in A$ (otherwise, it would have been $A_{x+1}$).
The set of edges $\{A_{x+1},\ldots,A_{c}\}$ form a perfect matching  of $G_x$ so there must be two edges and two indices $l_1>x$ and $l_2 > x$, such that $A_{l_1}=(a,d), A_{l_2}=(b,c)$. We assume w.l.o.g. that $l_1<l_2$.
The edge $(a,b)$ is of maximum weight in $G_x$ therefore $w(a)\geq w(c)$ and $w(b)\geq w(d)$. It follows that $w(c,d) \leq w(a,b) \leq t$, and therefore $(c,d)\in G_x$.
Let $A'=A \setminus \{(a,d),(b,c)\} \cup \{(a,b),(c,d)\}$.
Clearly, $A'$ is a perfect matching in $G$, $A'_{x+1}=(a,b)$ and therefore $A'$ shares a prefix of length $x+1$ with $Q$.
If $A'$ dominates $A$, then since $Q$ does not dominate $A$, it follows that $Q$ does not dominate $A'$. Thus $A'$ is a perfect matching that shares a prefix of length $x+1$ with $Q$ and is not dominated by $Q$. This is a  contradiction to the choice of $A$.  We finish the proof by showing that $A'$ dominates $A$.

Let $l_3$ be the index such that $A'_{l_3}=(c,d)$. Since $w(b)\geq w(d)$, $l_3>l_2$.
Let $\Delta(l)= \sum\limits_{i=1}^{l} w(A'_{i}) - \sum\limits_{i=1}^{l} w(A_{i})$.
The matchings $A'$ and $A$ share a prefix of length $x$, so for every $1\leq l \leq x$, $\Delta(l)=0$. For $x+1 \leq l < l_1$, $\Delta(l) = w(a,b) - w(A_{l}) \geq 0$ since $(a,b)$ is the edge of maximum weight in $G_x$.
For $l_1 \leq l < l_2$, $\Delta(l) = w(a,b) - w(a,d)\geq 0$ also by the maximality $(a,b)$.
For $l_2 \leq l < l_3$, $\Delta(l) = w(A'_{l}) - w(c) - w(d)$ which is non-negative because $l<l_3$ and therefore $w(A'_{l}) \geq w(A'_{l_3})= w(c)+w(d)$.
For $l\geq l_3$, $\Delta(l)=0$.
This shows that $A'$ dominates $A$ and concludes our proof that $Q$ is a dominant perfect matching in $G$.
\end{proof}

\subsubsection{On dominant perfect matchings in $d$-uniform hypergraphs}

The problem of finding a dominant perfect matching in a $d$-uniform threshold hypergraph\footnote{ A \textit{$d$-uniform threshold hypergraph} is defined on a set of vertices, $V$, each with a non-negative weight $w(v)$. The set of edges, $E$, contains all the subsets $S\subset V$ of size $d$ such that the sum of the weights of the vertices in $S$ is at most some fixed threshold $t>0$. } that has a perfect matching is interesting in the context of the single load, minimal cache version of the joint cache partition and job assignment problem. If we can find such a matching then an algorithm similar to Algorithm $A_{\frac{4}{3}}$ in Section \ref{slmc_43} would give a solution that uses $(d+1)K$ cache and approximates the makespan up to a factor of $\frac{d+2}{d+1}$.

However, the following example shows that in a $3$-uniform threshold hypergraph that has a perfect matching,  a dominant perfect matching does not necessarily exist.
 Let $\epsilon > 0$ be an arbitrarily small constant. Consider a hypergraph with 12 vertices, 3 vertices of each weight in $\{\frac{1}{3}, \frac{2}{9},\frac{4}{9}-\epsilon,\epsilon\}$. Each triplet of vertices is an edge if the sum of its weights is at most $1$. This hypergraph has a perfect matching. In fact, let's consider two perfect matchings in this hypergraph.
Matching $A$ consists of the edges $(\frac{1}{3},\frac{1}{3},\frac{1}{3})$, $(\frac{4}{9}-\epsilon, \frac{4}{9}-\epsilon, \epsilon)$, $(\frac{4}{9}-\epsilon, \frac{2}{9},\frac{2}{9})$ and $(\frac{2}{9},\epsilon,\epsilon)$.
Matching $B$ consists of three edges of the form $(\frac{1}{3},\frac{2}{9},\frac{4}{9}-\epsilon)$ and one edge of the form $(\epsilon,\epsilon,\epsilon)$.
It is easy to check that $A$ and $B$ are valid perfect matchings in this hypergraph.
Any dominant perfect matching in this hypergraph must contain the edge $(\frac{1}{3},\frac{1}{3},\frac{1}{3})$ in order to dominate $A$, since this is the only edge of weight $1$ in this hypergraph.
The sum of the two heaviest edges in matching $B$ is $2-2\epsilon$ and therefore any dominant perfect matching must have an edge of weight at least $1-2\epsilon$, as otherwise the matching will not dominate matching $B$. But, if the edge $(\frac{1}{3},\frac{1}{3},\frac{1}{3})$ is in the dominant matching, then all edges disjoint from $(\frac{1}{3},\frac{1}{3},\frac{1}{3})$ have a weight smaller than $1-2\epsilon$. Thus no dominant perfect matching exists in this hypergraph.

Matching $A$ in the example above is the perfect matching found by applying the greedy algorithm to this hypergraph.
It is interesting to note that in a $3$-uniform threshold hypergraph, the greedy algorithm does not necessarily find a perfect matching at all.
This is because Lemma \ref{not_stuck} does not extend to $3$-uniform threshold hypergraphs.
Let $\epsilon > 0$ be an arbitrarily small constant.
Consider a hypergraph with 9 vertices, 3 vertices of each weight in $\{\frac{1}{3}, \frac{2}{9},\frac{4}{9}-\epsilon\}$.
Each triplet of vertices is an edge if the sums of its weights is at most $1$.
This hypergraph has a perfect matching since the 3 edges of the form $(\frac{1}{3},\frac{2}{9},\frac{4}{9}-\epsilon)$ are a perfect matching in this hypergraph.
However the greedy algorithm first selects the edge $(\frac{1}{3},\frac{1}{3},\frac{1}{3})$ and then selects an edge of the form $(\frac{2}{9},\frac{2}{9},\frac{4-\epsilon}{9})$.
The remaining hypergraph now contains three vertices and no edges, so the greedy algorithm is stuck and fails to find a perfect matching.

\subsection{PTAS for jobs with correlative single load and minimal cache demand}\label{correl_ptas}

The main result in this section is a polynomial time approximation scheme for instances of the single load minimal cache demand problem, where there is a correlation between the load and the cache demand of jobs with non-zero cache demand. This special case is motivated by the observation that often there is some underlying notion of a job's ``hardness'' that affects both its load and its minimal cache demand.

Consider an instance of the single load minimal cache demand problem such that for any two jobs $j,j'$ such that $x_j$ and $x_{j'}$ are non-zero, $a_j \leq a_{j'} \iff x_j \leq x_{j'}$.
We call a job $j$ such that $x_j>0$ a \textit{demanding job} and a job $j$ such that $x_j=0$ a \textit{non-demanding job}.
We consider the following decision problem:
We want to decide if there is a cache partition of $K$ cache to $c$ cores and an assignment of jobs to the cores such that the job's minimal cache demand is satisfied and that the resulting makespan is at most $m$? By scaling down the loads of the jobs by $m$, we assume w.l.o.g that $m=1$.

Let $\epsilon>0$. We present an algorithm that if there is a cache partition and a job assignment with makespan at most $1$, returns a cache partition and a job assignment with makespan at most $(1+2\epsilon)$.
Otherwise, our algorithm either decides that there is no solution of makespan at most $1$ or returns a solution of makespan at most $(1+2\epsilon)$.
Combining this algorithm with a binary search, we obtain a PTAS.

If there is a job $j$ such that $a_j>1$ then our algorithm decides that there is no solution of makespan at most $1$.
Thus we assume that for any $j$, $a_j\leq 1$.

Let $J=J_1\cup J_2$, $J_1=\{j\in J \mid a_j \geq \epsilon\}$, $J_2=J\backslash J_1$.
In the first phase, we deal only with jobs in $J_1$.
For each $j\in J_1$ let $u_j=\max\{u\in\mathbb{N} \mid \epsilon + u\epsilon^2 \leq a_j \}$.
We say that $\epsilon + u_j\epsilon^2$ is the \textit{rounded-down load} of job $j$.

Let $U_D=\{u_j \mid j\in J_1,\, x_j>0\}$ and $U_{ND}=\{u_j \mid j\in J_1,\, x_j=0\}$.
An \textit{assignment pattern} of a core is a table that indicates for each $u\in U_D$ how many demanding jobs of rounded-down load $\epsilon + u\epsilon^2$ are assigned to the core and for each $u\in U_{ND}$ how many non-demanding jobs of rounded-down load $\epsilon + u\epsilon^2$ are assigned to the core.
Note that an assignment pattern of a core does not identify the actual jobs assigned to the core.
We only consider assignment patterns whose rounded-down load is at most $1$.

A \textit{configuration of cores} is a table indicating how many cores we have of each possible assignment pattern.
A configuration of cores $T$ is \textit{valid} if for every $u\in U_D$, the number of demanding jobs in $J_1$ whose $u_j=u$ equals the sum of the numbers of demanding jobs with $u_j=u$ in all assignment patterns in $T$ and, similarly, for every $u\in U_{ND}$, the number of non-demanding jobs in $J_1$ whose $u_j=u$ equals the sum of the numbers of non-demanding jobs with $u_j=u$ in all assignment patterns in $T$.

The outline of our algorithm is as follows.
The algorithm enumerates over all valid configurations of cores.
For each valid configuration $T$, we find an actual assignment of the jobs in $J_1$ that matches $T$ and minimizes the total cache used.
We then proceed to assign the jobs in $J_2$, in a way that guarantees that if there a solution of makespan $1$ and $K$ cache that matches this configuration of cores, then we obtain a solution of makespan at most $(1+2\epsilon)$ and at most $K$ cache. If our algorithm does not generate a solution of makespan at most $(1+2\epsilon)$ and at most $K$ cache, for all valid configurations of cores, then our algorithm decides that no solution of makespan at most $1$ exists.

Let $T$ be a valid configuration of cores.
For each core $i\in\{1,\ldots,c\}$, let $q_i$ be the maximal rounded-down load of a demanding job assigned to core $i$ according to the assignment pattern of core $i$ in $T$.
Let $\alpha_i$ be the number of demanding jobs of rounded-down load $q_i$ on core $i$, according to $T$.
We assume w.l.o.g that the cores are indexed such that $q_i \geq q_{i+1}$.
Let $Q=\{q_i \mid i\in\{1,\ldots,c\}\}$.
For each $q\in Q$, let $s(q)$ be the index of the first core $i$ with $q_i=q$ and let $e(q)$ be the index of the last core $i$ with $q_i=q$.
Assume that the cores $s(q),\ldots,e(q)$ are indexed such that $\alpha_{s(q)}\geq \ldots \geq \alpha_{e(q)}$.
Let $J_{1}(q)=\{j\in J_1 \mid x_j\neq 0, \epsilon+u_j \epsilon^2 = q\}$, the set of all demanding jobs in $J_1$ whose rounded down load is $q$.
Let $Y(q)$ be the set of the $\sum\limits_{i=s(q)}^{e(q)} \alpha_{i} $ jobs of smallest cache demands in $J_{1}(q)$.

Our algorithm builds an assignment matching $T$ of minimal cache usage among all assignments matching $T$.
To do so, our algorithm goes over $Q$ in a decreasing order and distributes the jobs in $Y(q)$ to the cores $s(q),\ldots,e(q)$ in this order of the cores such that core $i \in [s(q),e(q)]$, in turn, gets the $\alpha_i$ most cache demanding jobs in $Y(q)$ that are not yet assigned.
After we assign the demanding jobs with the maximal rounded-down load on each core, our algorithm arbitrarily chooses the identity of all other jobs in the configuration $T$.
These are non-demanding jobs and demanding jobs whose rounded-down load is not of the maximal rounded-down load on their core.
Each core is allocated cache according to the cache demand of the most cache demanding job that is assigned to it.

The algorithm continues with the jobs in $J_2$. It first assigns the
demanding jobs in $J_2$, in the following greedy manner. Order these
jobs from the most cache demanding to the least cache demanding. For
each core, we consider two load values: its \textit{actual load}
which is the sum of the actual loads of jobs in $J_1$ assigned to
the core, and its \textit{rounded down load} which is the sum of
rounded down loads of jobs in $J_1$ assigned to the core. We order
the cores such that first we have all the cores that already had
some cache allocated to them in the previous phase of the algorithm,
in an arbitrary order. Following these cores, we order the cores
with no cache allocated to them, from the least loaded core to the
most loaded core, according to their rounded down loads. These cores
are either empty or have only non demanding jobs, from $J_1$,
assigned to them. The algorithm assigns the jobs to the cores in
these orders (of the jobs and of the cores) and stops adding more
jobs to a core and moves to the next one when the core's actual load
exceeds $1+\epsilon$. After all these jobs are assigned, the
algorithm adjusts the cache allocation of the cores whose most cache
demanding job is now a job of $J_2$.

Finally, it assigns the non-demanding jobs in $J_2$.
Each such job is assigned arbitrarily to a core whose actual load does not already exceed $1+\epsilon$.

\begin{lemma}\label{ptas_poly}
The number of  valid configurations of cores is $O(c^{O(1)})$.
\end{lemma}
\begin{proof}
We first consider the number of assignment patterns with rounded-down load at most $1$.
Since for each job $j$, $a_j\leq 1$, the size of $U_D$ and the size of $U_{ND}$ are at most $\left\lfloor\frac{1-\epsilon}{\epsilon^2}\right\rfloor=O(\frac{1}{\epsilon^2})=O(1)$.
In an assignment pattern of load at most $1$, there are at most $\frac{1}{\epsilon}$ jobs in $J_1$ assigned to each core and thus we get that the number of possible assignment patterns is at most $O((\frac{1}{\epsilon})^{(\frac{1}{\epsilon^2})})=O(1)$.
Since the number of assignment patterns we consider is $O(1)$, it follows that the number of possible configurations of cores is $O(c^{O(1)})$.
\end{proof}

Since our algorithm spends a polynomial time per configuration of cores then Lemma \ref{ptas_poly} implies that our algorithm runs in polynomial-time.

\begin{lemma}\label{take_smallest}
For any configuration of cores $T$ there is an assignment matching $T$ of minimal cache usage among all assignments matching $T$, that for each $q\in Q$ assigns the $\sum\limits_{i=s(q)}^{e(q)} \alpha_{i}$ least cache demanding jobs in $J_1(q)$ (i.e. the set of jobs Y(q)) to the cores $s(q),\ldots,e(q)$.
\end{lemma}
\begin{proof}
Consider a job assignment $S$ of minimal cache usage that matches $T$.
Assume that for some $q\in Q$ assignment $S$ does not assign all the jobs in $Y(q)$ to the cores $s(q),\ldots,e(q)$.
So there is a core $i\in [s(q),e(q)]$ that runs a job $j\in J_1(q)\setminus Y(q)$.

Since $S$ assigns $\sum\limits_{i=s(q)}^{e(q)} \alpha_{i}$ jobs from $J_1(q)$ to cores $s(q),\ldots, e(q)$ and since jobs in $J_1(q)$ cannot be assigned to cores $i'>e(q)$, it follows that there is a core $i'<s(q)$ and a job $j'\in Y(q)$ such $S(j')=i'$.
Suppose we switch the assignment of jobs $j$ and $j'$ and run job $j$ on core $i'$ and job $j'$ on core $i$. Let $S'$ denote the resulting assignment.
The cache required by core $i'$ does not increase, as it runs demanding jobs of rounded down load greater than $q$ and therefore of cache demand greater than the cache demand of job $j$.
By the choice of the jobs $j$ and $j'$ we know that $x_{j'} \leq x_{j}$ and therefore the cache required by core $i$ in $S'$ can only decrease compared to the cache required by core $i$ in $S$.
It follows that the cache usage of $S'$ is at most that of $S$ and since $S$ is of the minimal cache usage of all assignments that match $T$, we get that the cache usage of $S'$ must be the same as of $S$.

By repeating this argument as long as there is a job that violates Lemma \ref{take_smallest}, we obtain an assignment as required.
\end{proof}

\begin{lemma}\label{by_order}
For any configuration of cores $T$, Let $S$ be an assignment matching $T$ such that  for each $q\in Q$ and for each core $i\in[s(q),e(q)]$,
if we index the jobs in $Y(q)$ from the most cache demanding to the least cache demanding, assignment $S$ assigns to core $i$ the jobs in $Y(q)$ of indices $\sum\limits_{j=s(q)}^{i-1} \alpha_j +1,\ldots, \sum\limits_{j=s(q)}^{i} \alpha_j$.
Assignment $S$ is of minimal cache usage, among all assignments matching $T$.
\end{lemma}
\begin{proof}
Assume to the contrary that assignment $S$ is not of minimal cache usage, among all assignments matching $T$.
Let $S'$ be an assignment whose existence is guaranteed by Lemma \ref{take_smallest}.
Since $S$ and $S'$ have different cache usages, there exists $q\in Q$ such that $S$ and $S'$ differ on their assignment of the jobs in $Y(q)$.
We index the jobs in $Y(q)$ from the most cacn demanding to the least cache demanding.
Let $j\in Y(q)$ be the first job (most cache demanding) in $Y(q)$ such that $S(j)\neq S'(j)$.
We select $S'$ such that it maximizes $j$ among all assignments satisfying Lemma \ref{take_smallest} that disagree with $S$ on the assignment of the jobs in $Y(q)$.

Denote $i=S(j)$ and $i'=S'(j)$.
Since $S$ and $S'$ both assign $\alpha_i$ jobs from $Y(q)$ to core $i$ and since $j$ is the first job in $Y(q)$ on which $S$ and $S'$ disagree, then there is a job $j_2\in Y(q)$, $j_2>j$ such that $S'(j_2)=i$.

We first assume that there is a job $j_1<j$ such that $S(j_1)=i$.
Let $S''$ be the assignment such that $S''(j)=i$, $S''(j_2)=i'$ and for any job $h\not \in \{j_,j_2\}$, $S''(h)=S'(h)$.
The cache required by core $i'$ in $S''$ is at most the cache required by core $i'$ in $S'$, since $j<j_2$.
Since $j_1<j$ and $S(j_1)=i$, we know that $S'(j_1)=i$ and also $S''(j_1)=i$.
This implies that in $S''$, core $i$ requires the same amount of cache as in $S'$.
It follows that $S''$ is also an assignment of minimal cache usage, and that it satisfies Lemma \ref{take_smallest}. Since $S''(j)=S(j)$, we get a contradiction to the way we selected $S'$.
Thus $S$ is  of minimal cache usage, among all assignments matching $T$.

We now assume that $j$ is the first job in $Y(q)$ such that $S(j)=i$.
Let $S''$ be the following assignment.
Any job that is assigned by $S'$ to a core different than $i$ and $i'$ is assigned by $S''$ to the same core.
For any job $x$ such that $S'(x)=i'$,  $S''(x)=i$.
All the $\alpha_{i'}$ least cache demanding jobs assigned by $S'$ to core $i$ are assigned by $S''$ to core $i'$.
Note that $\alpha_i \geq \alpha_{i'}$ and therefore assignment $S''$ is well defined.

Since $S$ and $S'$ agrees on the assignment of jobs $\hat{j} < j$ in $Y(q)$ and assign them to cores $l<i$, then job $j$ is the most cache demanding job assigned to cores $l\geq i$ by $S'$ and $S''$.
Therefore in assignment $S'$, core $i'$ requires $x_j$ cache and in assignment $S''$ core $i$ requires $x_j$ cache.
In assignment $S''$, core $i'$ is assigned a set of jobs that is a subset of the jobs assigned to core $i$ by $S'$.
Thus the cache required by core $i'$ in assignment $S''$, is at most the cache required by core $i$ in assignment $S'$.
It follows that $S''$ is also an assignment of minimal cache usage, and that it satisfies Lemma \ref{take_smallest}.
This contradicts the choice of $S'$ and concludes the proof that assignment $S$ is of minimal cache usage, among all assignments matching $T$.
\end{proof}

\begin{corollary}\label{actual_opt_ident}
For each configuration of cores $T$ our algorithm builds an actual assignment of minimal cache usage of the jobs in $J_1$ that matches $T$.
\end{corollary}
\begin{proof}
The assignment returned by our algorithm is an assignment $S$, as in the statement of Lemma \ref{by_order}.
\end{proof}

\begin{lemma}\label{ptas_aprx_dcsn}
Consider an instance of the correlative single load minimal cache demand problem.
If there is a cache partition and job assignment that schedules the jobs on $c$ cores, uses at most $K$ cache and has a makespan of at most $1$ then our algorithm finds a cache partition and job assignment that schedules the jobs on $c$ cores, uses at most $K$ cache and has a makespan of at most $(1+2\epsilon)$.
\end{lemma}
\begin{proof}
Let $A$ be a solution of makespan at most $1$ with $c$ cores and $K$ cache, whose existence is assumed by the lemma.
Let $T_A$ be the configuration of the cores corresponding to the assignment of the jobs in $J_1$ by solution $A$ and assume our algorithm currently considers $T_A$ in its enumeration.

We show that our algorithm succeeds in assigning all the jobs to $c$ cores. Let's assume to the contrary that it fails. It can only fail if all cores are assigned an actual load of more than $(1+\epsilon) $ and there are still remaining jobs to assign. This indicates that the total volume to assign is larger than $c(1+\epsilon)$, which contradicts the fact that assignment $A$ is able to assign the jobs to $c$ cores with makespan at most $1$.

Let $S$ denote the assignment of all jobs on $c$ cores that out algorithm returns when it considers $T_A$. 
We know that $S$ matches  $T_A$ for jobs in $J_1$.
We now show that in $S$ each core has an actual load of at most $1+2\epsilon$.
When we restrict $S$ to $J_1$ we know that the rounded down load on each core is at most $1$ and that each core has at most $\frac{1}{\epsilon}$ jobs from $J_1$ assigned to it.
Since the actual load of any job in $J_1$ is at most $\epsilon^2$ larger than its rounded down load, we get that if we restrict assignment $S$ to $J_1$, the actual load on each core is at most $1+ \frac{\epsilon^2}{\epsilon} = 1+ \epsilon$.
The way our algorithm assigns the jobs in $J_2$ implies that the actual load of a core in assignment $S$ can only exceed $1+\epsilon$ by the load of a single job from $J_2$. Therefore the actual load on any core in assignment $S$ is at most $1+2\epsilon$.

We show that assignment $S$ uses at most $K$ cache.
Cache is allocated by our algorithm in two steps: when it decides on the actual assignment of the jobs in $J_1$ that matches $T_A$ and when it assigns the demanding jobs in $J_2$.
Lemma \ref{actual_opt_ident} shows that $S$ restricted to $J_1$ is of minimal cache usage of all assignments matching $T_A$ and thus uses at most the same amount of cache as assignment $A$ restricted to $J_1$.

We show that when we also take into account the demanding jobs in $J_2$, $S$ uses at most the same amount of cache as $A$.
Assume the cores in $S$ are indexed according to the order in which our algorithm assigns demanding jobs from $J_2$ to them.
Assume the cores in $A$ are indexed such that core $i$ in $S$ and core $i$ in $A$ have the same assignment pattern.
For any core in $S$, we say that its \textit{free space} is $(1+\epsilon)$  minus the sum of the actual loads of all jobs in $J_1$ assigned to it by $S$.
For any core in $A$, we say that its free space is $1$ minus the sum of the actual loads of all jobs in $J_1$ assigned to it by $A$.
For any $i$, core $i$ in $S$ has the same rounded down load as core $i$ in $A$ and the actual load of core $i$ in $S$ is at most $\epsilon$ larger than the actual load of core $i$ in $A$. 
Therefore, by the definition of free space, the free space of core $i$ in solution $S$ is at least the free space of core $i$ in solution $A$.

Let $i_2$ be the number of cores in $S$ that have a demanding job from $J_1$ assigned to them.
When our algorithm assigns jobs in $J_2$ to a core $i\leq i_2$,  it does not increase the cache required by core $i$ since any job in $J_1$ is at least as cache demanding as any job in $J_2$.
It follows that the total cache required by cores $1,\ldots,i_2$ in $S$ is at most the total cache required by cores $1,\ldots,i_2$ in A. 

Let $i>i_2$ be a core in $S$  whose cache demand is determined by a job from $J_2$.
We now show that core $i$ in $S$ requires no more cache than core $i$ in $A$. This will conclude the proof that $S$ uses at most $K$ cache. 

The total load of demanding jobs in $J_2$ that $S$ assigns to cores $1,\ldots,i-1$ is at least the sum of the free space of these cores, since our algorithm exceeds an actual load of $1+\epsilon$ on each core before moving the next. The sum of the free space of cores $1,\ldots,i-1$ in $S$ is at least the sum of the free space of the cores $1,\ldots,i-1$ in $A$, which in turn is an upper bound on the total load of demanding jobs from $J_2$ that are assigned in $A$ to cores $1,\ldots,i-1$.
Since our algorithm assigns the demanding jobs in $J_2$ in a non-increasing order of their cache demand we get that the cache demand of the most cache demanding job from $J_2$ on core $i$ in $S$ is at most the cache demand of the most cache demanding job in $J_2$ on core $i$ in $A$. 
\end{proof}

Lemma \ref{ptas_aprx_dcsn} shows that for any $\epsilon'>0$, we have a polynomial time $(1+2\epsilon')$-approximate decision algorithm.
Given $\epsilon>0$, by applying our algorithm with $\epsilon'=\epsilon/2$ we obtain a polynomial time $(1+\epsilon)$-approximate decision algorithm.

By using a binary search similar to the one in Lemma \ref{bin_opt} we obtain an $(1+\epsilon)$-approximation for the optimization problem, using our $(1+\epsilon)$-approximate decision algorithm.
To conclude, we have proven the following theorem.

\begin{theorem}
There is a polynomial time approximation scheme for the joint cache partition and job assignment problem, when the jobs have a correlative single load and minimal cache demand.
\end{theorem}

\section{Step functions with a constant number of load types}\label{sec_const}

Empirical studies \cite{Drepper} suggest that the the load of a job, as a function of available cache, is often similar to a step-function.
The load of the job drops at a few places when the cache size exceeds the working-set required by some critical part.
In between these critical cache sizes the load of the job decreases negligibly with additional cache. The problems we consider in this section are motivated by this observation.

Formally, each job $j\in J$ is described by two load values $l_j< h_j$ and a \textit{cache demand} $x_j\in\{0,\ldots,K\}$. If job $j$ is running on a core with at least $x_j$ cache then it takes $l_j$ time and otherwise it takes $h_j$ time. If a job is assigned to a core that meets its cache demand, $x_j$, we say that it is \textit{assigned as a small job}. If it is assigned to a core that doesn't meet its cache demand we say that it is \textit{assigned as a large job}.
At first we study the case where the number of different load types is constant and then we show a polynomial time scheduling algorithm for the corresponding special case of the ordered unrelated machines scheduling problem.

Let $L=\{ l_j \mid j \in J\}$ and $H=\{h_j \mid j\in J\}$, the sets of small and large loads, respectively.  Here we assume that $|L|$ and $|H|$ are both bounded by a constant.

For each $\alpha\in L$, $\beta\in H$, we say that job $j$ is of \textit{small type $\alpha$} if $l_j=\alpha$ and we say that job $j$ is of \textit{large type $\beta$} if $h_j=\beta$. If job $j$ is of small type $\alpha$ and large type $\beta$ we say that it is of \textit{load type $(\alpha,\beta)$}. Note that jobs $j_1,j_2$ of the same load type may have different cache demands $x_{j_1}\neq x_{j_2}$ and thus if we take cache demands into account the number of different job types is $\Omega(K)$ and not  $O(1)$.

We reduce this problem to the single load minimal cache demand problem studied in Section \ref{sec_slmc}.
For each load type $(\alpha,\beta)$, we enumerate on the number, $x(\alpha,\beta)$, of the jobs of load type $(\alpha,\beta)$ that are assigned as small jobs. For each setting of the values $x(\alpha,\beta)$ for all load types, we create an instance of the single load minimal cache demand problem in which each job corresponds to a job in our original instance.
For each job $j$ which is one of the $x(\alpha,\beta)$ most cache demanding jobs of load type $(\alpha,\beta)$ we create a job of load $\beta$ and cache demand $0$. For each job $j$ of load type $(\alpha,\beta)$ which is not one of the $x(\alpha,\beta)$ most cache demanding job of this load type, we create a job of load $\alpha$ and cache demand $x_j$.
We solve each of the resulting instances using any algorithm for the single load minimal cache demand problem presented in Section \ref{sec_slmc}, and choose the solution with the minimal makespan. We transform this solution back to a solution of the original instance, by replacing each job with its corresponding job in the original instance. Note that this does not affect the makespan or the cache usage.

\begin{lemma}
Given a polynomial time $\alpha$-approximation algorithm for the single load minimal cache demand problem that uses at most $\beta K$ cache, the reduction described above gives a polynomial time $\alpha$-approximation algorithm for the problem  where job loads are step functions with a constant number of load types, that uses at most $\beta K$ cache.
\end{lemma}
\begin{proof}
Consider an instance of the joint cache partition and job assignment problem with load functions that are step functions with a constant number of load types.
Assume there is a solution $A$ for this instance of makespan $m$ that uses at most $K$ cache. Let $x(\alpha,\beta)$ be the number of jobs of load type $(\alpha,\beta)$ that are assigned in $A$ as large jobs. W.l.o.g we can assume that that for each $(\alpha,\beta)$, the $x(\alpha,\beta)$ jobs that are assigned as large jobs are the $x(\alpha,\beta)$ most cache demanding jobs of load type $(\alpha,\beta)$. The existence of $A$ implies that when our algorithm considers the same values for $x(\alpha,\beta)$, for each $(\alpha,\beta)$, it generates an instance of the single load cache demand problem that has a solution of makespan at most $m$ and at most $K$ cache.
Applying the $\alpha$-approximation algorithm for the single load minimal cache demand problem, whose existence in assumed by the lemma, on this instance yields a solution of makespan at most $\alpha m$ that uses at most $\beta K$ cache. This solution is transformed to a solution of our original instance without affecting the makespan or the cache usage.

Our algorithm runs in polynomial time since the size of the enumeration is $O(n^{|L||H|})$.
\end{proof}

\begin{corollary}
For instances in which the load functions are step functions with a constant number of load types there are polynomial time approximation algorithms that approximate the makespan up to a factor of $2$, $\frac{3}{2}$ and $\frac{4}{3}$ and use at most $K$, $2K$ and $3K$, respectively.
\end{corollary}

\subsection{The corresponding special case of ordered unrelated
machines}\label{apn_const} Recall that if we fix the cache partition
in an instance of the joint cache partition and job assignment
problem then we obtain an instance of the ordered unrelated machines
scheduling problem. For the case where the load functions are step
functions with a constant number of load types, the resulting
ordered unrelated machines instance can be solved in polynomial time
using the dynamic programming algorithm described below. The dynamic
program follows a structure similar to the one used in \cite{leah},
where polynomial time approximation schemes are obtained for several
variants of scheduling with restricted processing sets.

In this special case of the ordered unrelated scheduling problem job $j$ runs in time $l_j$ on some prefix of the machines, and in time $h_j$ on the suffix (we assume that the machines are ordered in non-increasing order of their strength/cache allocation).  For simplicity, we assume $x_j$ is given as the index of the first machine on which job $j$ has load $h_j$. If job $j$ takes the same amount of time to run regardless of cache, we assume $x_j=c+1$ and its load on any machine is $l_j$. As before, we assume that $L=\{ l_j \mid j \in J\}$ and $H=\{h_j \mid j\in J\}$ are of constant size.

We design a polynomial time algorithm that finds a job assignment that minimizes the makespan.
The algorithm does a binary search for the optimal makespan, as in Section \ref {apn_approx_des_to_opt}, using an algorithm for the following decision problem: Is there an assignment of the jobs $J$ to the $c$ machines with makespan at most $M$?
By scaling the loads, we assume that $M=1$.

For every machine $m$, we define $S_m=\{j\in J\;\mid\;x_j=m+1 \}$, the set of all jobs that are large on machine $m+1$ and small on any machine $i\leq m$.
Let $S_m(\alpha,\beta)=\{j\in S_m \mid l_j=\alpha, \, h_j=\beta \}$ and $b_m(\alpha,\beta)=|S_m(\alpha,\beta)|$. It is convenient to think of $b_m$ as a vector in $\{0,\ldots,n\}^{L\times H}$.

Let $a \in \{0,\ldots,n\}^{L\times H}$, $\delta \in \{0,\ldots,n\}^{H}$ and $m$ be any machine.
Let $J(m,a)$ be a set of jobs which contains all the jobs in $\bigcup_{i=1}^m S_i$ together with additional $a(\alpha,\beta)$ jobs of load type $(\alpha,\beta)$ from $\bigcup\limits_{i=m+1}^{c} S_i$, for each load type $(\alpha,\beta)$.
Let $\pi_m(a,\delta)$ be $1$ if we can schedule all the jobs in $J(m,a)$, except for $\delta(\beta)$ jobs of each large load type $\beta$, on the first $m$ machines. Note that since the additional jobs specified by $a$ are small on all machines $1,\ldots,m$, $\pi_m(a,\delta)$ does not depend on the additional jobs' identity.
Our original decision problem has a solution if and only if $\pi_c(\vec{0},\vec{0})=1$.

Consider the decision problem  $\pi_1(a, \delta)$.
We want to decide if it is possible to schedule the jobs in $J(1,a)$, except for $\delta(\beta)$ jobs of each large load type $\beta$, on machine $1$. To decide this, our algorithm chooses the $\delta(\beta)$  jobs of each large job type $\beta$ that have the largest small loads and removes them from $J(1,a)$.
If the sum of the small loads of the remaining jobs is at most $1$,  then $\pi_1(a,\delta)=1$, and otherwise $\pi_1(a,\delta)=0$.

To solve $\pi_m(a,\delta)$ we enumerate, for each load type $(\alpha,\beta)$, on $\xi(\alpha,\beta)$, the number of jobs in $J(m,a)$ of this load type that are assigned as small jobs to machine $m$.
Note that these jobs are either in $S_m(\alpha,\beta)$ or in the additional set of $a(\alpha,\beta)$ jobs of type $(\alpha,\beta)$.
For each $\beta\in H$, we enumerate on the number $\lambda(\beta)$ of jobs in $J(m,a)$ of large load type $\beta$ that are assigned as large jobs to machine $m$.
The following lemma is the basis for our dynamic programming scheme. Its proof is straightforward.

\begin{lemma}\label{cond_lemma}
We can schedule all the jobs in $J(m,a)$ except for $\delta(\beta)$ jobs of large load type $\beta$ (for each $\beta\in H$) on machines $1,\ldots, m$ with makespan at most $1$ such that $\xi(\alpha,\beta)$ jobs of load type $(\alpha,\beta)$ are assigned to machine $m$ as small jobs and $\lambda(\beta)$ jobs of large load type $\beta$ are assigned to machine $m$ as large jobs if and only if the following conditions hold:

\begin{itemize}
\item For each $(\alpha,\beta)\in L\times H,\,  \xi(\alpha,\beta) \leq a(\alpha,\beta) + b_m(\alpha,\beta)$: The number of jobs of each load type that we assign as small jobs to machine $m$ is at most the number of jobs in $J(m,a)$ of this load type that are small on machine $m$.
\item $\sum\limits_{\beta\in H} \lambda(\beta) \beta \;+\;\sum\limits_{(\alpha,\beta)\in L\times H} \xi(\alpha,\beta) \alpha \;\leq 1$. The total load of the jobs assigned to  machine $m$ is at most $1$.
\item Let $a' = a + b_m  - \xi$ and $\delta'=\delta+\lambda$ then
$\pi_{m-1}(a',\delta')=1$.  The jobs in $J(m-1,a')$, except for $\delta'(\beta)$ jobs of large load $\beta$ for each $\beta\in H$,  can be scheduled on machines $1,\ldots,m-1$ with makespan at most $1$.
\end{itemize}

\end{lemma}

The algorithm for solving $\pi_m(a,\delta)$  sets $\pi_m(a,\delta)=1$ if it finds $\lambda$ and $\xi$ such that the conditions in  Lemma \ref{cond_lemma} are met. If the conditions are not met for all $\lambda$ and $\xi$ then $\pi_m(a,\delta)=0$.

Our dynamic program solves $\pi_m(a,\delta)$ in increasing order of $m$ from $1$ to $c$ and returns the result of  $\pi_{c}(\vec{0},\vec{0})$.
The correctness of the dynamic program follows from Lemma \ref{cond_lemma} and from the fact that for $m=1$, our algorithm chooses the jobs that it does not assign to machine $1$ such that the remaining load on machine $1$ is minimized. Therefore we set $\pi_1(a,\delta)=1$ if and only if there is a solution of makespan at most $1$.

By adding backtracking links, our algorithm can also construct a schedule with makespan at most $1$.
We maintain links between each $\pi_m(a,\delta)$ that is $1$ to a corresponding $\pi_{m-1}(a',\delta')$ that is also $1$, according to the last condition in Lemma \ref{cond_lemma}.
Tracing back the links from $\pi_{c}(\vec{0},\vec{0})$ gives us an assignment with makespan at most $1$ as follows.
Consider a link between $\pi_m(a,\delta)$ and  $\pi_{m-1}(a',\delta')$. This defines $\lambda= \delta' - \delta$ and $\xi = a + b_m - a'$.
For each $(\alpha,\beta)$ we assign to machine $m$,  $\xi(\alpha,\beta)$ arbitrary jobs of load type $(\alpha,\beta)$ from $\bigcup\limits_{i=m}^{c} S_i$ that we have not assigned already, and we reserve $\lambda(\beta)$ slots of load $\beta$ on machine $m$ to be populated with jobs later.
Our algorithm guarantees that the load on machine $m$ is at most $1$. When we reach $\pi_1(a,\delta)$, for some $a$ and $\delta$, in the backtracking phase, we have $\delta(\beta)$ slots of size $\beta$ allocated on machines $2,\ldots,m$.
The $\delta(\beta)$ jobs of large load $\beta$ with the largest small loads in $J(1,a)$ are assigned to these slots. Note that these jobs may be large on their machine and have a load of $\beta$ or they may be small and have a load smaller than $\beta$. In any case, the resulting assignment assigns all the jobs in $J$ and has a makespan of at most $1$.

The number of problems $\pi_m(a,\delta)$ that our dynamic program solves is $O(c n^{|L||H|})=O(c n^{O(1)})$.
To solve each problem, we check the conditions in Lemma \ref{cond_lemma} for $O(n^{|L||H|})$ possible $\lambda$'s and $\xi$'s. This takes $O(1)$ per $\lambda$ and $\xi$ since we already computed $\pi_{m-1}(a',\delta')$ for every $a'$ and $\delta'$. Thus the total complexity of this algorithm is polynomial. This concludes the proof of the following theorem.

\begin{theorem}
Our dynamic programming algorithm is a polynomial-time exact optimization algorithm for the special case of the ordered unrelated machines scheduling problem, where each job $j$ has load $l_j$ on some prefix of the machines, and load $h_j\geq l_j$ on the corresponding suffix.
\end{theorem}

\section{Joint dynamic cache partition and job scheduling}\label{sec_variants}

We consider a generalization of the joint cache partition and job assignment problem that allows for dynamic cache partitions and dynamic job assignments.
We define the generalized problem as follows.
As before, $J$ denotes the set of jobs, there are $c$ cores and a total cache of size $K$.
Each job $j\in J$ is described by a non-increasing function $T_j(x)$.

A dynamic cache partition $p=p(t,i)$ indicates the amount of cache allocated to core $i$ at time unit $t$\footnote{To simplify the presentation we assume that time is discrete.}.
For each time unit $t$, $\sum\limits_{i=1}^{c} p(t,i) \leq K$.
A dynamic assignment $S=S(t,i)$ indicates for each core $i$ and time unit $t$, the index of the job that runs on core $i$ at time $t$. If no job runs on core $i$ at time $t$ then $S(t,i)=-1$. If $S(t,i)=j\neq -1$ then for any other core $i_2\neq i$, $S(t,i_2)\neq j$.
Each job has to perform 1 \textit{work unit}. If job $j$ runs for $\alpha$ time units on a core with $x$ cache, then it completes $\frac{\alpha}{T_j(x)}$ work. A partition and schedule $p,S$ are \textit{valid} if all jobs complete their work.
Formally, $p,S$ are valid if for each job $j$, $\sum\limits_{<t,i>\in S^{-1}(j) } \frac{1}{T_j(p(t,i))} = 1$.
The \textit{load} of core $i$ is defined as the maximum $t$ such that $S(t,i)\neq -1$.
The makespan of $p,S$ is defined as the maximum load on any core. The goal is to find a valid dynamic cache partition and dynamic job assignment with a minimal makespan.

It is easy to verify that dynamic cache partition and dynamic job assignment, as defined above, generalize  the static partition and static job assignment.
The partition is static if for every fixed core $i$, $p(t,i)$ is constant with respect to $t$.
The schedule is a static assignment if for every job $j$, there are times $t_1<t_2$ and a core $i$ such that $S^{-1}(j)=\{<t,i> \mid t_1\leq t \leq t_2\}$.

We consider four variants of the joint cache partition and job assignment problem.
The static partition and static assignment variant studied so far, the variant in which the cache partition is dynamic and the job assignment is static, the variant in which the job assignment is dynamic and the cache partition is static and the variant in which both are dynamic.

Note that in the variant where the cache partition is dynamic but the job assignment is static we still have to specify for each core, in which time units it runs each job that is assigned to this core. That is, we have to specify a function $S(t,i)$ for each core $i$. This is due to the fact that different schedules of the same set of jobs assigned to a particular core, when the cache partition is dynamic, may have  different loads, since jobs may run with different cache allocations. When the cache partition is also static, the different schedules of the same set of jobs on a particular core have the same load, and it suffices to specify which jobs are assigned to which core.

We study the makespan improvement that can be gained by allowing a dynamic solution.
We show that allowing a dynamic partition and a dynamic assignment can improve the makespan by a factor of at most $c$, the number of cores.
We also show an instance where by using a dynamic partition and a static assignment we achieve an improvement factor arbitrarily close to $c$.
We show that allowing a dynamic assignment of the jobs, while keeping the cache partition static, improves the makespan by at most a factor of $2$, and that there is an instance where an  improvement of $2-\frac{2}{c}$ is achieved, for $c\geq2$.

Given an instance of the joint cache partition and job assignment problem, we denote by $O_{SS}$ the optimal static cache partition and static job assignment, by $O_{DS}$ the optimal dynamic cache partition and static job assignment, by $O_{SD}$ the optimal static cache partition and dynamic job schedule and by $O_{DD}$ the optimal dynamic cache partition and dynamic job schedule. For any solution $A$ we denote its makespan by $M(A)$.

\begin{lemma}\label{dsub}
For any instance of the joint cache partition and job assignment problem,
$M(O_{SS})\leq cM(O_{DD})$.
\end{lemma}
\begin{proof}
Let $A$ be the trivial static partition and schedule, that assigns all jobs to the first core and allocates all the cache to this core.
Let's consider any job $j$ that takes a total of $\alpha$ time to run in the solution $O_{DD}$.
Whenever a fraction of job $j$ runs on some core with some cache partition, it has at most $K$ cache available to it.
Therefore, in solution $A$, when we run job $j$ continuously on one core with $K$ cache,  it take at most $\alpha$ time.
Since the total running time of all the jobs in solution $O_{DD}$ is at most $cM(O_{DD})$, we get $M(O_{SS})\leq M(A) \leq c M(O_{DD})$.
\end{proof}

\begin{corollary}
For any instance of the joint cache partition and job assignment problem,
 $M(O_{SS})\leq cM(O_{DS})$.
\end{corollary}
\begin{proof}
Clearly, $M(O_{DS})\geq M(O_{DD})$ for any instance. Combine this with Lemma \ref{dsub} and we get that $M(O_{SS})\leq cM(O_{DS})$
\end{proof}

\begin{lemma}\label{dslb}
For any $\epsilon>0$ there is an instance of the joint cache partition and job assignment problem, such that $M(O_{SS})> (c-\epsilon) M(O_{DS})$.
\end{lemma}
\begin{proof}
Let $b$ be an arbitrary constant.
Let's consider the following instance with two types of jobs.
There are $c$ jobs of type $1$, such that for each such job $j$, $T_j(x)=\infty$, for $x<K$ and $T_j(K)=1$.
There are $bc$ jobs of type $2$, such that for each such job $j$, $T_j(x)=bc$ if $x<\frac{K}{c}$ and $T_j(x)=1$ if $x\geq \frac{K}{c}$.

Consider the following solution. The static job assignment runs $b$ jobs of type $2$ on each core. After $b$ time units, it runs the $c$ jobs of type $1$ on core $1$.
The dynamic cache partition starts with each core getting $\frac{K}{c}$ cache. The cache partition changes after $b$ time units and core $1$ gets all the cache.
This solution has a makespan of $b+c$ and therefore $M(O_{DS})\leq b+c$.

There is an optimal static cache partition and static job assignment that allocates to each core $0,\frac{K}{c}$ or $K$ cache, because otherwise we can reduce the amount of cache allocated to a core without changing the makespan of the solution.
This implies that there are only two static cache partitions that may be used by this solution optimal static solution: the partition in which $p(i)=\frac{K}{c}$ for each core $i$, and the partition that gives all the cache to a single core.
It is easy to see that if we use the cache partition where $p(i)=\frac{K}{c}$ we get a solution with an infinite makespan because of the jobs of type $1$.
Therefore this optimal static solution uses a cache partition that gives all the cache to a single core.
Given this partition, the optimal job assignment is to run all the $c$ jobs of type $1$ on the core with all the cache, and assign to that core additional  $bc-(c-1)$ jobs type $2$.
So the load on that core is $bc+1$.
Each of the $c-1$ cores with no cache is assigned exactly one job of type $2$, and each such core has a load of $bc$.
Therefore the ratio  $\frac{M(O_{SS})}{M(O_{DS})} \geq \frac{bc+1}{b+c}$. The lower  bound on this ratio approaches $c$ as $b$ approaches infinity.
Since $b$ is an arbitrarily chosen constant, we can choose it large enough such that we get a lower bound that is greater than $c-\epsilon$, for any $\epsilon>0$.
\end{proof}

\begin{corollary}
For any $\epsilon>0$ there is an instance of the joint cache partition and job assignment problem, such that $M(O_{SS})> (c-\epsilon) M(O_{DD})$.
\end{corollary}
\begin{proof}
Consider the same instance as in the proof of Lemma \ref{dslb}.
For that instance, $M(O_{SS})> (c-\epsilon) M(O_{DS})$.
It follows that  $M(O_{SS})> (c-\epsilon) M(O_{DD})$ for the instance in  Lemma \ref{dslb} , since $M(O_{DS})\geq M(O_{DD})$.
\end{proof}

\begin{lemma}\label{sdub}
For any instance of the joint cache partition and job assignment problem, $M(O_{SS})\leq 2M(O_{SD})$.
\end{lemma}
\begin{proof}
Consider any instance of the joint cache partition and job assignment problem and let $O_{SD}=(p,S)$.
Let $x_{ij}$ be the fraction of job $j$'s work unit that is carried out by core $i$. Formally, $x_{ij}= \frac{|\{t \mid (t,i)\in S^{-1}(j) \} |   }{T_j(p(i))}$.
Let's consider the instance of scheduling on unrelated machines where job $j$ runs on core $i$ in time $T_j(p(i))$.
Since for every job $j$, $\sum\limits_{i=1}^{c} x_{ij} = 1$ then $x_{ij}$ is a fractional assignment for that instance of the unrelated machines scheduling problem.
The makespan of this fractional solution is $M(O_{SD})$.
Let $y$ be the optimal fractional assignment of the defined instance of unrelated machines. We know that if we apply Lenstra's rounding theorem \cite{LST90} to $y$, we get an integral assignment for the unrelated machines scheduling instance, denoted by $z$, such that the makespan of $z$ is at most twice the makespan of $y$ and therefore at most twice the makespan of $x$.
Assignment $z$ is a static job assignment and therefore $(p,z)$ is a solution to the joint static cache partition and static job assignment problem of our original instance, with makespan at most twice $M(O_{SD})$. It follows that $M(O_{SS})\leq 2M(O_{SD})$.
\end{proof}

\begin{lemma}\label{sdlb}
For $c\geq 2$, there is an instance of the joint partition and scheduling problem such that $\frac{M(O_{SS})}{M(O_{SD})} = 2-\frac{2}{c}$.
\end{lemma}
\begin{proof}
Consider the following instance. There are $c$ jobs, where each takes  $1-\frac{1}{c}$ time regardless of the cache allocation, and one job that takes $1$ time unit, regardless of cache.
The optimal static schedule for this instance assigns two jobs of size $1-\frac{1}{c}$ to the first core, assigns one job of size $1-\frac{1}{c}$ to each of the cores $2,\ldots,c-1$, and assigns the unit sized job to the last core. This yields a makespan of $2-\frac{2}{c}$.
The optimal dynamic assignment assigns one job of size $1-\frac{1}{c}$ fully to each core, and then splits the unit job equally among the cores, to yield a makespan of $1$. Notice that this can be scheduled in a way the the unit job will never run simultaneously on more than one core. This is achieved by running the $i$th fraction of size $\frac{1}{c}$ of the unit job on core $i$ at time $\frac{i-1}{c}$. The other jobs, that are fully assigned to a single core, are paused and resumed later, if necessary, to accommodate the fractions of the unit sized job.
Therefore in this instance the ratio   $\frac{M(O_{SS})}{M(O_{SD})}$  is exactly $2-\frac{2}{c}$.
\end{proof}

\bibliography{references}

\end{document}